\documentclass[onecolumn]{IEEEtran}

\usepackage{mystyle}
\usepackage{booktabs}
\usepackage{dsfont,xcolor,graphicx,subcaption}
\usepackage{stmaryrd}
\usepackage[ruled,vlined]{algorithm2e}
\usepackage{cite}

\usepackage{xspace}
\usepackage{bbm}
\def\diag{\mathop{\rm diag}\nolimits}%
\newcommand{\Ac}{\mathcal{A}}

\newcommand{\Cc}{\mathcal{C}}
\newcommand{\Dc}{\mathcal{D}}
\newcommand{\Ec}{\mathcal{E}}

\newcommand{\Nc}{\mathcal{N}}

\newcommand{\Qc}{\mathcal{Q}}

\newcommand{\Xc}{\mathcal{X}}

\newcommand{\ex}{{\rm e}}

\newcommand{\ev}{{\bf e}}
\newcommand{\gv}{\bm{g}}
\newcommand{\Xv}{{\bf X}}

\newcommand{\xv}{{\bf x}}
\newcommand{\yv}{{\bf y}}
\newcommand{\zv}{{\bf z}}
\newcommand{\uv}{{\bf u}}

\newcommand{\sv}{{\bf s}}
\newcommand{\dv}{{\bf d}}

\newcommand{\thetav}{\bm{\theta}}

\newcommand{\Xt}{{\tilde{X}}}

\DeclareMathOperator\E{E}
\let\P\relax
\DeclareMathOperator\P{P}

\def\textiid{i.i.d.\@\xspace}
\newcommand\iid{\ifmmode\text{ i.i.d. } \else \textiid \fi}

\newcommand{\NN}{\mathbb{N}}

\title{Compressed sensing  in the presence of speckle noise}

\author{}

\newcommand{\av}{{\bf a}}

\newcommand{\wv}{{\bf w}}
\newcommand{\cv}{{\bf c}}

\newcommand{\xvh}{\hat{\xv}}
\newcommand{\xvt}{\tilde{\xv}}

\renewcommand{\baselinestretch}{1.5}

\newtheorem{corollary}{Corollary}
\newtheorem{remark}{Remark}

\begin{document}

\author{Wenda Zhou, \and Shirin Jalali, \and Arian Maleki}

\author{Wenda Zhou,~\IEEEmembership{}
  Shirin~Jalali,~\IEEEmembership{}
Arian Maleki,~\IEEEmembership{}
                   \thanks{W. Zhou is with the NYU’s Center for Data Science, New York, NY (e-mail: wz2247@nyu.edu).}
           \thanks{S. Jalali is with Nokia Bell Labs, Murray Hills, NJ  (e-mail: shirin.jalali@nokia-bell-labs.com).}
           \thanks{A. Maleki is with the Statistics Department of Columbia University, New York, NY. (e-mail: arian@stat.columbia.edu).}}

\maketitle
\begin{abstract}
The problem of recovering a structured signal from its linear measurements in the presence of speckle noise is studied. This problem appears in many imaging systems such as synthetic aperture radar and optical coherence tomography. The current acquisition technology oversamples  signals and converts the problem into a denoising problem with multiplicative noise. However, this paper explores the possibility of reducing the number of measurements below the ambient dimension of the signal. The sophistications that appear in the study of multiplicative noises have so far impeded  theoretical analysis of such problems. This paper aims to present the first theoretical result regarding the recovery of signals from their undersampled measurements under the speckle noise. It is shown that if the signal class is structured, in the sense that the signals can be compressed efficiently, then one can obtain accurate estimates of the signal from fewer measurements than the ambient dimension. We demonstrate the effectiveness of the methods we propose through simulation results.  
\end{abstract}

\section{Introduction}
\subsection{Problem statement}
Various modern imaging methods, such as synthetic aperture radar (SAR) \cite{moreira2013tutorial} and optical coherence tomography (OCT) \cite{huang1991optical}, rely on coherent imaging. The main drawback of such coherence-based imaging systems is that they all suffer from a granular noise that is typically referred to as speckle (or multiplicative)  noise. At a high level, the reason for observing such a noise is the following. In  applications of such imaging methods, surfaces reflecting the incoming coherent waves are all rough, when considered  at a resolution  comparable to the wavelength of the signal. Therefore, the reflected rays will have different phases, which means that at detection point, reflected rays corresponding to some pixels will  add up constructively, while reflected points corresponding to other pixels add up destructively. This phenomenon can be modeled as a multiplicative noise. 

  In this paper, our goal is to develop a theoretical framework for studying the imaging problem in the presence of speckle noise. We focus on the problem of compressed sensing, i.e., recovering a structured signal from its underdetermined measurements in the presence of such noise.  Different  speckle noise filtering methods (for the case that the measurement matrix is identity) have been proposed in the literature over the past couple of decades. However, to the best of our knowledge,  there  has not been  any theoretical  analysis of speckle noise and corresponding optimal recovery methods when fewer measurements than the ambient dimension are available. As will be discussed later in the paper, compared to the well-studied traditional additive noise, which appears in other imaging systems such as magnetic resonance imaging, multiplicative noise poses many more theoretical and practical challenges. In this paper, we discuss and address some of these challenges. 
    
Let ${\cal Q}$ denote a compact subset of $\mathds{R}^n$ that describes the class of structured signals we are interested in, e.g., class of natural images or more abstractly, class of bounded $k$-sparse signals   or class of bounded piece-wise-constant signals. For $\xv\in\mathds{R}^n$, define  $X=\diag(\xv)=\diag(x_1,\ldots,x_n)$. In imaging systems with  speckle noise, the measurement vector $\yv\in\mathds{R}^m$ is defined as 
\[
\yv=AX\wv+\zv.
\]
Here, the multiplicative noise $\wv\in\mathds{R}^n$ and additive noise $\zv\in\mathds{R}^m$ are i.i.d.~$\Nc(0,\sigma_w^2)$ and i.i.d.~$\Nc(0,\sigma_z^2)$, respectively.\footnote{Note that in real-world systems, the speckle noise is typically  complex-valued and has a Gaussian distribution. The signal and the measurement matrix are also complex-valued. However, for notational simplicity we have focused on real-valued signals and noise here. Extension of the results to complex-valued signals and noises is straightforward. } The difference between such a system and a standard linear measurement system is clearly in the multiplicative noise $\wv$ that distorts each input pixel independently. Note that in such  systems, unlike the additive noise $\zv$, the dimensions of the multiplicative noise depends on the dimensions of the input signal $\xv$ and not  the number of measurements $m$. Given such a measurement process, there are various fundamental  questions one can ask:
\begin{enumerate}
\item Is accurate recovery feasible in the presence of speckle noise?
\item What is the performance of maximum likelihood-based (ML-based) recovery method that takes the source structure into account?
\item Given a class of structured signals $\Qc$ with a certain intrinsic dimension, what is the minimum number of measurements $m$ that guarantees the feasibility of accurate recovery?
\end{enumerate}
In this paper, we aim to address these questions. First, inspired by  maximum likelihood estimation, we derive a compression-based ML recovery method that employs  compression codes designed for our desired class of signals ($\Qc$) to define, capture and exploit  the source structure. We characterize the performance of the derived compression-based ML recovery method and prove that given enough number of measurements (related to the desired accuracy and rate-distortion performance of the compression code), it is able to recover the signal within desired accuracy, with probability approaching one, as the dimensions of the problem grow without bound. 


\subsection{Related work}

As mentioned earlier, speckle noise is an inherent problem  in coherence-based imaging systems. The current technology in such imaging systems is to collect at least as many measurements as  the ambient dimension of the signal, and convert the problem into a denoising problem, i.e. recovering the signal $\xv$ from 
\[
\yv=\xv\odot \wv + \zv,
\]
where $\odot$ denotes element-wise multiplication, i.e., $y_i=x_iw_i$, and $\wv$ and $\zv$ denote iid Gaussian noises.    

For each type of such imaging systems, various denoising methods have been developed over the years. One such type of imaging that is most relevant to the measurement model studied in this paper is SAR imaging \cite{moreira2013tutorial}.  The classical techniques for denoising speckle noise in SAR imaging have been reviewed in \cite{argenti2013tutorial} and \cite{touzi2002review}. In  \cite{ozcan2016despeckling}, the authors propose a sparsity-based  total-variation  approach for smoothing and denoising speckled images. Application of  non-local means approaches to speckle noise reduction are explored in \cite{deledalle2014patch,dimartino2016nonlocal}.  In recent years, inspired by the success of deep neural networks,  especially convolutional neural nets, in solving various inference tasks involving images, researchers have  explored application of such tools in SAR speckle noise denoising as well \cite{chierchial2017convolutional,wang2017convolutional}. 

One of the well-known applications of compressed sensing is  inverse SAR (ISAR) imaging. Compressed sensing enables such systems to produce images of equal or even higher quality in  shorter acquisition time compared to conventional systems. This is achieved by requiring much lower number of samples  \cite{yoon2008cs,patel2009cs}. Various challenges faced in  applications of compressed sensing to ISAR imaging have been explored in the literature. (Refer to  \cite{onhnon2010sparsity, demirci2013cs, cheng2018pareto, bi2017multifrequency,cetin2014sparsity} for a noncomprehensive list of  such works.) However,    compressed sensing in the presence of speckle noise is not addressed directly in the literature, and unlike  conventional setups, the issue of speckle noise, while still present, is ignored. In this paper, we address this gap and study the problem of compressed sensing in the presence of speckle noise. We  take a theoretical approach to the problem and derive a compression-based ML compressed sensing recovery method in the presence of speckle noise. We characterize the performance of the derived solution for the case that  the additive noise  approaches zero. 

\subsection{Notations and definitions}

Sets are denoted by calligraphic letters.
The size of a set $\Ac$ is denoted by $|\Ac|$.
The $\ell_2$-norm of  an $m\times n$ matrix $A$ is defined as $\|A\|=\max_{\xv\neq 0} \|A\xv\|_2/\xv\|_2$.
The Hilbert-Schmidt (or Frobenius) norm of $A$ is $\|A\|_{\rm HS}=(\sum_{i,j}a_{i,j}^2)^{0.5}$.
Throughout the paper, $\log$ refers to natural logarithm.
For a vector $\xv\in\mathds{R}^n$, ${\rm diag}(\xv)$ denotes the $n\times n$ diagonal matrix with diagonal elements determined by $\xv$.
For $x\in\mathds{R}$, $b$-bit quantized version of $x$ is defined as $[x]_b=2^{-b}\lfloor2^b x \rfloor$.

\subsection{Paper organization}
To show the effect of multiplicative noise and how using the source structure can improve the recovery performance,  Section \ref{sec:denoise} reviews a simple denoising problem in the presence of multiplicative noise. Section \ref{sec:structured-signals} reviews compressible codes and motivates using
compression codes to define and enforce  structure of signals. Section \ref{sec:ML-recovery} characterizes the log likelihood function corresponding to the problem of compressed sensing in the presence of multiplicative noise and derives some of its properties.  \ref{sec:compress} introduces  an optimization recovery based on  compression-based ML   for compressed sensing
in the presence of speckle noise is derived. The performance of the proposed optimization is also analyzed. Section
\ref{sec:inference} discusses application of projected gradient descent for
approximating the solution of the optimization corresponding to ML recovery. The
proofs and simulation results are presented in Section \ref{sec:proof} and
Section \ref{sec:sim}, respectively. Section \ref{sec:conc} concludes the paper.


\section{An illustrative example}\label{sec:denoise}

To illustrate how inference in the presence of multiplicative noise can be done, we start with a simple denoising example.  Let 
\[
\yv=\xv\odot \wv,
\]
where $\odot$ denotes element-wise multiplication, i.e., $y_i=x_iw_i$. Assume that $\wv\stackrel{\rm i.i.d.}{\sim}\Nc(0,1)$.  Note that if we consider $x_i=u_i$, then $y_i\sim\Nc(0,u_i^2)$. Therefore, ignoring the constant terms, the log likelihood function is equal to
\[
\log \prod_{i=1}^n{1\over u_i}\exp(-{y_i^2\over 2u_i^2})=\sum_{i=1}^n\Big(-\log |u_i|-{y_i^2\over 2 u_i^2}\Big). 
\]
As expected, since the noise has a symmetric distribution around zero, the log-likelihood function is symmetric with respect to $u_i$. In other words, the sign of $u_i$ cannot be inferred from the given measurements. Hence, in the rest of this section we assume that all the $u_i$s are positive. The log likelihood can be optimized over the individual signal components $u_i$. The function $-\log u_i-{y_i^2\over 2 u_i^2}$ is a quasi-concave function of $u_i$ and  maximizing the likelihood function leads to $\xvh_{\rm ML}=|\yv|$. Hence, the mean square error of the maximum likelihood estimate is given by
\begin{equation}\label{eq:power_denoising}
\E[ \|\xvh_{\rm ML}-\xv\|_2^2] = 2\|\xv\|_2^2  (1 - \sqrt{\frac{2}{\pi}}). 
\end{equation}
Now, we would like to consider a class of structured signals and see the effect of the structure on the MSE. Consider a very simple structure for the signal. Let $\Qc=\{a{\bf 1}_n:\; a\in(x_m,x_M)\}$, where ${\bf 1}_n=[1,\ldots,1]^T\in\mathds{R}^n$.  That is,  $\Qc$ denotes constant $n$-dimensional vectors with a value in $(x_m,x_M)$. Because of the sign ambiguity issue mentioned earlier,  we assume that $0<x_m<x_M$. To take the  known structure of the source into account, we could maximize the likelihood function over the set of constant signals. Let $\uv=\alpha {\bf 1}_n$. Then, the likelihood function simplifies to
\[
-n\log \alpha-{1\over 2\alpha^2}\sum_{i=1}^ny_i^2,
\]
which is again a quasi-concave function in terms of $\alpha$. Maximizing the likelihood function leads to
\[
\xvh_{\rm ML}=\hat{\alpha}_{\rm ML} {\bf 1}_n, \;\;\;{\rm where}\;\; \hat{\alpha}_{\rm ML}=({1\over n}\sum_{i=1}^ny_i^2)^{1\over 2}.
\]
Given $\xv=a{\bf 1}_n$ and $\yv=\xv\odot \wv$, we have
\begin{align*}
\hat{\alpha}_{\rm ML}=\Big({a^2\over n}\sum_{i=1}^nw_i^2 \Big)^{1\over 2}=a \Big({1\over n}\sum_{i=1}^nw_i^2\Big)^{1\over 2}.
\end{align*}
Now, suppose that $n$ is large. Then the calculations in the appendix show that 
\[
\E[ \|\xvh_{\rm ML}- \xv\|_2^2] = \E [n (\hat{\alpha}_{\rm ML}-a)^2] = a^2 n \E[ ( ({1\over n}\sum_{i=1}^nw_i^2)^{1\over 2})-1)^2 ]= a^2 n (2- 2 \E  ({1\over n}\sum_{i=1}^nw_i^2)^{1\over 2}) \rightarrow \frac{a^2}{2}. 
\]
Therefore, in this case,  ML estimation that takes the signal structure into account asymptotically recovers the underlying signal much more accurately than the one that does not use the information about the signal  structure. There are a few points that we would like to emphasize here:

\begin{enumerate}
\item The likelihood in the multiplicative noise problem is more complicated than the likelihood in additive noise models. While the likelihood is quasi-convex in the denoising problem we discussed in this section, in general the likelihood can become non-convex in the sensing problem. 

\item As expected, given that the noise is always present in such systems, the exact recovery is never possible. As can be seen in our simple example, the best that can be expected from such systems is that the normalized mean square error $\frac{1}{n}{\rm MSE} = O(\frac{1}{n})$. Note that the structure of the signal, i.e. the fact that the signal is constant, was the major help in reducing the MSE. Otherwise, the MSE of the MLE would be proportional to the power of the signal as shown in \eqref{eq:power_denoising}. We will clarify the notion of structure that will be used in this paper in the next section. 

\end{enumerate}


\section{Structured signals}\label{sec:structured-signals}
In this paper, we study the problem of compressed sensing in the presence of speckle noise. For $\xv\in\mathds{R}^n$, define  $X=\diag(\xv)=\diag(x_1,\ldots,x_n)$. In imaging systems with the speckle noise, the measurement vector $\yv\in\mathds{R}^m$ is defined as 
\[
\yv=AX\wv+\zv.
\]
Here, $\wv\in\mathds{R}^n$ and $\zv\in\mathds{R}^m$ are i.i.d.~$\Nc(0,\sigma_w^2)$ and i.i.d.~$\Nc(0,\sigma_z^2)$, respectively. The main goal of this paper is to study the problem of recovering the vector $\xv$ from measurements $\yv$, under the assumption that the number of measurements $m<n$. To solve this undersampled problem,  the underlying signal has to be structured. Otherwise, there are infinitely many signals that satisfy the measurement constraints, even without any noise in the system. Furthermore, as illustrated in the denoising example in Section \ref{sec:denoise}, the structure can potentially improve the performance of the recovery algorithms. 

In this paper, we define the signal structure based on a family of compression algorithms. This method that was introduced and developed in a series of papers \cite{jalali2016compression, beygi2019efficient, bakhshizadeh2020using, RezagahJ:17}, has the following advantages over the popular sparsity structure: 

\begin{enumerate}
\item Compression-based compressed sensing naturally expands the scope of compressed sensing algorithms for a class of signals from sparsity in a transform domain to that of compression codes designed for the same class. The latter is much richer and contains sparsity as its special case. 
\item  Given that the state-of-the-art compression algorithms, such as JPEG-2000, use sophisticated structures in the data, compression-based compressed sensing algorithms can potentially use such sophisticated structures for the signal recovery. Hence, compression-based algorithms that treat off-the-shelf compression codes as blackboxes that both define and enforce the source structure achieve very competitive performance in various applications such as compressed sensing of natural images \cite{beygi2019efficient} and phase retrieval \cite{bakhshizadeh2020using}.

\item In cases where the minimum required sampling rate ($m/n$) is known, ideal compression-based compressed sensing methods theoretically achieve those bounds \cite{RezagahJ:17}.

\end{enumerate}
In the rest of this section, we clarify the way we consider signal structure based on a family of compression codes. 

Let ${\cal Q}$ denote a compact subset of $\mathds{R}^n$ representing a class of structured signals. We assume that this class is structured in the sense that a family of compression codes indexed by their rate is known for ${\cal Q}$. More formally, we have access to a family of encoder mappings $\Ec_r:\mathds{R}^n \to \{1,2,\ldots, 2^{nr}\}$ and decoder mappings $\Dc_r:\{1,2,...,2^{nr}\}\to\mathds{R}^n$ indexed with rate $r$. For the signal $\xv \in {\cal Q}$, $\Ec_r(\xv)$ denotes its compressed version  (with $nr$ bits), while $\Dc_r(\Ec_r(\xv))$ denotes the reconstructed version of $\xv$ based on the $nr$ bits which were available to the compression algorithm. To see a concrete example, we can consider ${\cal Q}$ as the class of natural images of certain size and let $\Ec_r$ and $\Dc_r$ denote the compression and decompression algorithms of JPEG or JPEG2000 at a given rate $r$. Given that the compression algorithm is usually lossy in real-world applications, we define the distortion of our compression code as  
\[
\delta_r=\sup_{\xv\in\Qc}{1\over n} \|\xv-\Dc_r(\Ec_r(\xv))\|^2.
\]
Naturally, it is expected that  the distortion $\delta_r$  decreases as the rate $r$ increases. The rate-distortion function of the given family of lossy compression codes is defined as
\[
r(\delta) = \inf\{r  \ : \ \delta_r \leq \delta\}. 
\]
Finally, any given encoder and decoder mappings $(\Ec,\Dc)$ define a codebook 
\[
\Cc=\{\Dc(\Ec(\xv)):\; \xv\in\Qc\}
\] 
of size as most $2^{nr}$. The codebook represents all the fixed points of the mapping $\Dc_r(\Ec_r(\cdot))$. Intuitively, according to our compression code, the elements of the codebook are the simplest signals because they can be exactly reconstructed from $2^{nr}$ bits. 

Note that one can construct lossy compression codes for any class of signals ${\cal Q}$, for instance by simply quantizing the elements of every $\xv\in \Qc$. However, if the class ${\cal Q}$ is structured, then for a given compression rate $r$, there exist algorithms that achieve lower distortion compared to those that perform element-wise quantization. The following example further clarifies this point. 
Let $\Qc$ denote $[0,1]^n$. As is clear, the signals in this class do not have a particular structure. Hence, a standard compression algorithm for this class cannot do better than standard quantization. Suppose that we quantize each element of $\xv$ at resolution $\epsilon$. It is straightforward to see that the rate and distortion of this compression code become
\begin{eqnarray}
\delta &=&  \epsilon^2, \nonumber\\
r &=& \frac{1}{n} \log\left( \frac{1}{\epsilon} \right)^n = \log\left( \frac{1}{\epsilon} \right) . 
\end{eqnarray}
Hence, $r_Q(\delta) = 0.5\log \left( \frac{1}{\delta} \right)$, where the subscript $\Qc$ only denotes that this rate is calculated for the set $\Qc$. (The base of all the logarithms is $2$.) $r_Q(\delta) = 0.5\log \left( \frac{1}{\delta} \right)$ is our baseline as it shows the rate-distortion function of a compression algorithm on an unstructured class of signals.  
To compare this rate distortion function with an achievable rate-distortion function for a class of structured signals, consider the set 
\[
\tilde{\Qc} = \{ \xv \in \Qc \ : \ x_{k+1} = x_{k+2}= \ldots = x_n=0 \}.
\]
It is straightforward to see that the same coder that we used for $Q$ offers the following rate-distortion performance on $\tilde{\Qc}$
\[
r_{\tilde{\Qc}}(\delta) = \frac{k}{2n} \log  \left( \frac{k}{n\delta} \right)
\]
 Note that if the sparsity level $k$ is much smaller than the ambient dimension $n$, then  $\frac{r_{\tilde{\Qc}} (\delta)}{r_{\Qc}(\delta)} \ll 1$. In the rest of the paper, whenever we mention that a class of signals $\Qc$ is structured, it is assumed that $r_{\Qc}(\delta) \ll  \log \left( \frac{1}{\delta} \right)$.

\section{Characterizing  and understanding the log likelihood function}\label{sec:ML-recovery}
 

Our goal is to recover $\xv\in\mathds{R}^n$ from  measurements  
\begin{equation}\label{eq:linear_speckle}
\yv=AX\wv+\zv,
\end{equation}
where, $\wv\in\mathds{R}^n$ and $\zv\in\mathds{R}^m$ are i.i.d.~$\Nc(0,\sigma_w^2)$ and $\Nc(0,\sigma_z^2)$, respectively. Given the complexity of the model, in this section, we ignore the structure of $\xv$, and only study the log-likelihood function of \eqref{eq:linear_speckle}. In the next section, we will use the results we derive in this section to obtain results for the problem of structured signal recovery. 

 To compute the likelihood $p(\yv|X;A)$, ignoring the terms that do not depend on $X$, we have
\begin{align}
p(\yv|X;A)&= \int p(\yv|\wv,X;A)p(\wv)d\wv\nonumber\\
&\propto \int \exp(-{1\over 2\sigma_z^2}\|\yv-AX\wv\|^2-{1\over 2\sigma_w^2}\|\wv\|^2)d\wv\nonumber\\
&\propto \int \exp\Big(-{1\over 2} \wv^T({1\over \sigma_w^2}I_n+ {1\over \sigma_z^2}XA^TAX)\wv+ {1\over \sigma_z^2}\yv^TAX\wv\Big)d\wv\nonumber\\
&\propto |\Sigma|^{1\over 2}\exp({1\over 2}\mu^T\Sigma^{-1}\mu),
\end{align}
where $\Sigma$ and $\mu$ satisfy
\[
\Sigma^{-1}={1\over \sigma_w^2}I_n+ {1\over \sigma_z^2}XA^TAX,
\] 
and 
\[
\mu^T\Sigma^{-1}={1\over \sigma_z^2}\yv^TAX.
\]
Therefore, again ignoring the terms not depending on $X$, the log-likelihood function can be written as
\begin{align}
2\ell(X)
&=\log\det \Sigma +\mu^T\Sigma^{-1}\mu\nonumber\\
&=-\log\det({1\over \sigma_w^2}I_n+ {1\over \sigma_z^2}XA^TAX)+{1\over \sigma_z^4}\yv^TAX({1\over \sigma_w^2}I_n+ {1\over \sigma_z^2}XA^TAX)^{-1}XA^T\yv.\label{eq:def-ell-X}
\end{align}
Since the main focus of this paper is on the speckle noise, we next consider the case where the power of additive noise converges to zero ($\sigma_z\to 0$) and simplify the likelihood function accordingly. Depending on the number of measurements, the limit as $\sigma_z \rightarrow 0$ changes. Hence, we calculate the limit in two separate cases:

\begin{itemize}
\item {Case I: $AA^T$ is invertible ($m<n$).}\label{sec:case1-ML}
The following theorem simplifies the log-likelihood function for the small additive noise. Note that here we have assumed that the number of measurements $m$ is less than the ambient dimension $n$. 

 \begin{theorem}\label{thm:A-At-invert}
Assume that $AA^T$ is invertible. Then, as $\sigma_z\to 0$, $\ell(X)$ converges to
\begin{align}\label{eq:mlsimp_low}
\ell(X)&=-{1\over 2}\log\det(AX^2A^T)-{1\over 2\sigma_w^2}\yv^T(AX^2A^T)^{-1}\yv.
\end{align}
\end{theorem}

The proof of this theorem can be found in Section \ref{proof_thm:A-At-invert}. Note that the dependence of the log-likelihood on the signal $\xv$ is not as simple as the log-likelihood of the additive noise. This likelihood comes out of matrix manipulations and algebra, and it is not clear why the likelihood has this form.  To get a better intuition,  consider the following {\em hypothetical} acquisition model where the measurements have the following form:
\begin{equation}\label{eq:mixedmeasurements}
\yv= AX_o \mathbf{1}_n + AX_o\wv+\zv,
\end{equation}
The only difference here compared to the original problem is that now we have an extra term of the form $AX_o \mathbf{1}_n$ in the measurement. We should emphasize again that we are not aware of any real-world acquisition system for which the model in \eqref{eq:mixedmeasurements} is accurate. However, as will be discussed later, this model enables us to better understand some of the terms that appear in \eqref{eq:mlsimp_low}. Using the model in \eqref{eq:mixedmeasurements}, and following the same strategy as the one we used in Theorem \ref{thm:A-At-invert} for $\sigma_z \rightarrow 0$, it is straightforward to show that the corresponding likelihood function converges to 
\begin{align}
2\ell (X) = -\log\det(AX^2 A^T) - \frac{1}{\sigma_w^2} (\yv-A \bm{x})^T (AX^2A^T)^{-1} (\yv-A\bm{x}). \label{eq:likelihood-with-bias}
\end{align}

To gain some intuition behind this function, we now rederive it using some simplifying assumptions. Given that $\sigma_z \rightarrow 0$, the measurement model simplifies to $\yv= AX_o + AX_o\wv$. Interpreting  $ \uv= AX_o \wv$ as an additive noise and assuming  that (even though not accurate) the  noise is Gaussian and  independent of  with $AX_o$. It is straightforward to see that under these assumptions, $\uv \sim \Nc(0, \sigma_w^2 AX_o^2A^T)$. Now, if we write the log-likelihood for the model
\[
\yv= AX_o\mathbf{1}_n+ \uv, 
\]
with $\uv \sim \Nc(0, \sigma_w^2 AX^2A^T)$, we obtain \eqref{eq:likelihood-with-bias}. Note that the covariance matrix of $\uv$ also depends on the signal and should not be ignored in the likelihood. While the assumptions we have made in our heuristic argument are not accurate, they explain why for instance we should expect $(AX^2A^T)^{-1}$ to be multiplied by our measurements. Also, it explain where the logarithm term in our expressions comes from.

%
%


\item{Case II: $A^TA$ is  invertible ($m>n$).}\label{sec:denoising}
While the main focus of this paper is on the underdetermined settings where $m<n$, for the sake of completeness, for completeness, we also derive the log-likelihood for the $m>n$ case. For this case,  note that multiplying  both sides of $\yv=AX_o\wv$ with $A^T$, it follows that $X_o\wv=(A^TA)^{-1}A^T\yv$. In other words, in this case, we need to solve a denoising problem in the presence of a multiplicative noise.  But given $y=xw$, where $z\sim\Nc(0,\sigma_w^2)$ the likelihood of $y$ given $x$ is ${1\over \sqrt{2\pi x^2\sigma_w^2}}\ex^{-{y^2\over 2x^2\sigma_w^2}}$. Therefore, in this case the likelihood function of $\xv\in\mathds{R}^n$ can be derived in the following way. 

 \begin{theorem}\label{thm:At-A-invert}
Assume that $A^TA$ is invertible. Let ${\bf b}=(A^TA)^{-1}A^T\yv$. Then, as $\sigma_z\to 0$, $\ell(\xv)$ converges to
\begin{align}
\ell(\xv)&=-\sum_{i=1}^n\Big( {1\over 2}\log x_i^2+{b_i^2\over 2\sigma_w^2 x_i^2}\Big).
\end{align}
\end{theorem}
While we derived this formula by converting the inverse problem to a denoising problem, one can start with \eqref{eq:def-ell-X} and simplify it to reach the same conclusion.

\end{itemize}

 \section{Compressed sensing in the presence of speckle noise}\label{sec:compress}

\subsection{Our proposed recovery optimization problem} \label{ssec:recovery_opt}
Consider recovering $\xv_o\in\Qc\subset\mathds{R}^n$ from measurements $\yv=AX_o\wv$, where  $X_o=\diag(\xv_o)$ and $m<n$. Also, assume that $\xv_o\in\Qc\subset\mathds{R}^n$, where $\Qc$ is a set of structured signals for which we have a family of compression codes indexed with rate $r$ with corresponding codebook $\Cc_r$. Then, employing Theorem \ref{thm:A-At-invert}, when $m<n$, we use  the following  optimization problem to obtain an estimate of $\xv_o$:
\begin{align}
\hat{\xv}=\argmin_{\substack{X=\diag(\xv), \\ \xv\in\Cc_r}}\Big[ \log\det(AX^2A^T)+{1\over \sigma_w^2}\yv^T(AX^2A^T)^{-1}\yv\Big].\label{eq:main-opt-ML}
\end{align}
In other words, instead of minimizing the negative log-likelihood over all signals, we only focus on the signals that have simple representation according to our compression code, i.e. they can be represented exactly with $2^{nr}$ bits. Given that we have access to a family of compression codes, the rate $r$ can be considered as a free parameter that the user can tune to obtain better performance. We will discuss this parameter in more details later.  

Note that this optimization problem \eqref{eq:main-opt-ML} is hard not only because the fact that the negative log-likelihood is non-convex, but also because the set on which we would like to solve our optimization problem is discrete and very large. We will later discuss practical approaches to approximate the solution of this optimization problem.


\subsection{Main theoretical result}


As described in Section \ref{sec:denoising}, when the number of measurements exceeds the ambient dimension of the problem, i.e., $m\geq n$, the problem reduces to a denoising problem. Therefore, we focus on the case of $m<n$. 
Given $A\in\mathds{R}^{m\times n}$, define mapping $\Sigma: \mathds{R}^n\to \mathds{R}^{m\times m}$ as
\begin{align}
\Sigma(\xv)=(AX^2A^T)^{-1},
\end{align}
if the inverse is well-defined. (Here $X=\diag({\xv})$.) Also, define function $\mathds{R}^{m\times m}\to\mathds{R}$, as 
\[
f(\Sigma)=-\log\det \Sigma +{1\over \sigma_w^2}{\rm Tr}(\Sigma\yv\yv^T).
\]
Let $\yv=AX_o\wv$, where $X_o=\diag({\xv_o})$. Consider a rate-$r$ distortion-$\delta$ compression code with codebook $\Cc_r$. As argued in Section \ref{sec:case1-ML},  an ML-based recovery algorithm that takes advantage of the given compression code solves the following optimization problem:
\begin{align}
\hat{\xv}_o=\argmin_{\substack{X=\diag({\xv}), \\ \xv\in\Cc_r}} f(\Sigma(\xv)).\label{eq:min-f-main}
\end{align}
The main theoretical result of this paper is the following theorem, which shows that given sufficient number of measurements which is characterized in terms of the properties of the compression code, the described ML-based recovery method is able to recover $\xv$ from the measurements. To avoid sign ambiguity issue we discussed in Section \ref{sec:denoise} we assume that the elements of $\xv_o$ are all positive.

\begin{theorem}\label{thm:main-result}
Let $\yv=AX_o\wv$, where $X_o=\diag({\xv_o})$. Let $\xvh_o$ denote the solution of \eqref{eq:min-f-main}. Assume that $m<n/4$ and define $\gamma\triangleq ({1+2\sqrt{m/n}\over 1-2\sqrt{m/n}})^2$, $\alpha \triangleq {x_{\max}\over x_{\min}}$, $\rho_1\triangleq4\sqrt{2}\alpha^8\gamma^5(1+2\alpha^2\gamma)^2(1+2\sqrt{m/n})^2$ and $\rho_2\triangleq (1+2\alpha^2\gamma)^2\gamma^7\alpha^{14}$. 
Then, for any $\epsilon>0$, 
\begin{align}
 {1\over n} \| \xv_o- \xvh_o\|_2^2 \leq \rho_1 \sqrt{(1+\epsilon) nr  \over m} + \rho_2 x_{\max}^2\delta,
\end{align}
with a probability higher than
\[
1-n\ex^{-0.09m}-n\ex^{-0.84m}-2^{-nr \epsilon+1}-2\ex^{-{m\over 2}}.
\]
Note that $r$ and $\delta$ denote the rate and distortion of the compression algorithm used in \eqref{eq:min-f-main}. 
\end{theorem}

We would like to make a few remarks about this result.

\begin{remark}
Unlike the classical compressed sensing problem, here the exact recovery is not possible even though there is no additive noise in the measurements. This is due to two issues: 
\begin{enumerate}
\item The existence of multiplicative noise. Given the multiplicative noise that is present in the system, it should be clear that the exact recovery is not possible. 

\item The distortion in our compression algorithm: Given that our search space is only on the codewords, we should not expect to be able to get to the exact solution, because the exact solution may not be even a codeword. We expect this effect to diminish as $\delta \rightarrow 0$. As is clear the term $\rho_2 x_{\max}^2\delta \rightarrow 0$ as $\delta \rightarrow 0$. However, note that as we let $\delta \rightarrow 0$, we most probably are letting $r(\delta)$ grow to infinity, which in turn will blow up the first term in the MSE, i.e.  $\rho_1 \sqrt{(1+\epsilon) nr  \over m}$. Hence, the best choice of $\delta$ depends on the rate-distortion function. In practice, one may use a cross-validation technique for finding the best choice of $\delta$. 

\end{enumerate}
\end{remark}

\begin{remark}
 The condition $m<n/4$ does not allow us to let $m$ go to infinity without bound in our theorem. Note that as we discussed in Section \ref{sec:ML-recovery}, in order to derive the optimization problem \eqref{eq:min-f-main} we assumed that $m<n$. Hence, it is expected that solving  \eqref{eq:min-f-main} does not generate any reliable result for $m >n$. Furthermore, as $m$ gets close to $n$, the matrices $AX^2A^T$ start to have eigenvalues close to zero for many signals $X^2$. This seems to have some adverse algorithmic and statistical effect on the recovery of $\xv_o$. While in practice, we would always like to employ these systems in highly underdetermined regime, still a better understanding of the problem when $n/4<m<n$ will shed more light on the landscape and shape of the likelihood function in multiplicative noise systems. 
\end{remark}

Before we proceed to the discussion of the algorithmic issues of solving \eqref{eq:min-f-main}, we would like to clarify the statement of this theorem through a well-known example. In this example we assume that the family of compression codes on $\Qc$ satisfy $r(\delta) \leq \frac{k}{n} \log \frac{1}{\delta} + \frac{k}{n} \log (n)$. For instance, it is straightforward to construct such a family compression codes for the class of $k$-sparse signals in unit sphere, or class of piecewise constant signals with $k$ jumps.   

\begin{corollary}
Suppose that the family of compression algorithms satisfy $r(\delta) \leq \frac{k}{n} \log \frac{1}{\delta} + \frac{k}{n} \log n$. For any $\epsilon>0$, if we use the compression code with distortion $\delta= \frac{1}{n}$, then we have (the rest of the notations are the same as the ones introduced in Theorem \ref{thm:main-result})
\begin{align}\label{eq:error_ksparse}
 {1\over n} \| \xv_o- \xvh_o\|_2^2 \leq \rho_1 \sqrt{2(1+\epsilon)k \log n   \over m} + \rho_2 x_{\max}^2\frac{1}{n},
\end{align}
with a probability higher than
\[
1-n\ex^{-0.09m}-n\ex^{-0.84m}-2^{-2\epsilon k \log n +1}-2\ex^{-{m\over 2}}.
\]
\end{corollary}
Note that given the fact that $m< \frac{n}{4}$ the dominant term in the right hand side of \eqref{eq:error_ksparse} is the first term. The term $\frac{k \log n}{m}$ is the term that appears in the compressed sensing problem with additive noise as well \cite{bickel2009simultaneous}.  Note however that, in the additive noise case, instead of having $\sqrt{k \log n   \over m} $, we have ${k \log n   \over m}$. Whether this is an artifact of our proof technique or a fundamental difficulty of the multiplicative noise problem is an important open problem for our future research.


\section{Recovery algorithms}\label{sec:inference}

\subsection{Roadmap}
In Section \ref{sec:ML-recovery}, we showed that a compression-based ML recovery
method  recovers signal $\xv_o\in\Qc$ from measurements  $\yv=AX\wv_o$, where
$X_o=\diag({\xv_o})=\diag({x_{o,1}},\ldots,{x_{o,n}})$, by solving the
optimization described in \eqref{eq:main-opt-ML}. However, \eqref{eq:main-opt-ML} is a challenging
optimization as the cost function is a high-dimensional non-convex function of
$\xv$ over an exponentially large discrete set. Therefore, obtaining the solution of \eqref{eq:main-opt-ML} is not straightforward.

As an initial step in solving this optimization problem, we propose two different algorithms which attempt to solve \eqref{eq:main-opt-ML}:
(1) a projected gradient descent algorithm, and (2) a gradient-free multi-level optimization algorithm.
We describe the derivations of the algorithms in this section, and the simulation results in section~\ref{sec:sim}.

\subsection{Projected Gradient Descent Algorithm}\label{ssec:pgd}

\subsubsection{General algorithm}
Projected gradient descent (PGD) algorithm and its close relative proximal gradient descent \cite{combettes2005signal} are among the most popular algorithms in the fields of signal processing, machine learning, and optimization. In each step, PGD takes a small step in the direction of the gradient of the cost function (ignoring all the constraints on the solution) and then project the current estimate on the constraint set. These algorithms are known to converge to a global minimizer of convex optimization problems (if the step size is picked according to certain rules). PGD algorithms  have also been applied in practice to non-convex optimization problems. While they have exhibit a good performance for non-convex problems as well, there is still no general theory that can explain the good performance of PGD for a large class of non-convex problems. However, researchers have proved the success of such algorithms in particular instances of non-convex optimization problems \cite{blumensath2009iterative, beygi2019efficient, bakhshizadeh2020using}. In this paper, we derive the PGD algorithm for optimization problem \eqref{eq:main-opt-ML}. We also mention some details about how we implement the algorithm, and evaluate its performance through simulations. We leave the theoretical evaluation of this algorithm for a future research. As we discussed before, we are interested in solving the following optimization problem with PGD:
\begin{align*}
\hat{\xv}=\argmin_{\substack{X=\diag(\xv), \\ \xv\in\Cc_r}}\Big[ \log\det(AX^2A^T)+{1\over \sigma_w^2}\yv^T(AX^2A^T)^{-1}\yv\Big].
\end{align*}
The PGD algorithm has two steps: (i) moving in the direction of the gradient of $ \log\det(AX^2A^T)+{1\over \sigma_w^2}\yv^T(AX^2A^T)^{-1}\yv$, and (ii) projecting onto the constraint set $\Cc_r$. We discuss each step and some other details below.

\begin{itemize} 

\item Gradient calculation: To derive the PGD algorithm, we first need to compute the
gradient of the cost function $f(\Sigma(\xv))= \log\det (AX^2A^T) +{1\over \sigma_w^2}
\yv^T (AX^2A^T)^{-1}\yv$ with respect to $\xv$. Note that because $X$ is a diagonal  matrix, $AX^2A^T=\sum_{i=1}^nx_i^2
\av_i\av_i^T$, where $\av_i$ is the $i$-th column of matrix $A$.  Define
$\uv=\xv^2$, i.e., $u_i=x_i^2$. We first compute the partial derivative of
$\underline{f}(\uv)=\log\det (\sum_{i=1}^nu_i \av_i\av_i^T)+{1\over \sigma_w^2} \yv^T (\sum_{i=1}^nu_i \av_i\av_i^T)^{-1}\yv = f(\Sigma(\xv))$ with respect to $u_i$.
Let $\ev_i$ denote the unit vector in direction $i$.
Let $B = \sum_{j=1}^nu_j \av_j\av_j^T$.
Then,
\begin{align}
\underline{f}(\uv+\delta u_i \ev_i)-\underline{f}(\uv)&=\log\det (B +\delta_i \av_i\av_i^T )+{1\over \sigma_w^2} \yv^T (B +\delta_i \av_i\av_i^T )^{-1}\yv -\log\det B +{1\over \sigma_w^2} \yv^T B^{-1}\yv \nonumber\\
&=\log\det (I_m +\delta_i B^{-{1\over 2}} \av_i\av_i^TB^{-{1\over 2}} )+{1\over \sigma_w^2} \yv^T ((B +\delta_i \av_i\av_i^T )^{-1}-B^{-1})\yv.
\end{align}
Using the Woodbury matrix identity,
\[
(B +\delta_i \av_i\av_i^T )^{-1}=B^{-1}- {\delta_i B^{-1}\av_i\av_i^TB^{-1} \over 1 +\delta_i\av_i^T B^{-1}\av_j}.
\]
Also, given that all but one of the eigenvalues of  $I_m +\delta_i B^{-{1\over 2}} \av_i\av_i^TB^{-{1\over 2}} $ are equal to one, we have $\log\det (I_m +\delta_i B^{-{1\over 2}} \av_i\av_i^TB^{-{1\over 2}} )=\log(1+\delta_i \av_i^TB^{-1}\av_i)$. Using these identities, it follows that ${\partial \underline{f}(\uv) \over \partial u_i}= \av_i^TB^{-1}\av_i-(\av_i^TB^{-1} \yv)^2$, and 
\begin{equation}\label{eq:gradient}
{\partial f (\Sigma(\xv)) \over \partial x_i}= 2x_i(\av_i^TB^{-1}\av_i-(\av_i^TB^{-1} \yv)^2),
\end{equation}
where $B=AX^2A^T$. Equation \eqref{eq:gradient} gives us the gradients that are necessary for the PGD algorithm. 
 
\item Projection, proper step-size, and initialization:  As discussed before, PGD algorithm requires the projection onto $\Cc_r$ defined as
\[
    \pi_{\Cc_r}(\uv) = \argmin_{\uv' \in \Cc} \norm{\uv' - \uv}.
\]
 Note that, because $\Cc_r$ is often a very large non-convex set, this projection can be computationally demanding. While in some cases, such as the one we will mention in the next section, the projection can be calculated efficiently and accurately using for instance dynamic programming, in many other examples this is not the case. One successful approximation for $\pi_{\Cc_r}(u)$ that has shown promising results in other applications (e.g. \cite{beygi2019efficient, bakhshizadeh2020using}) is 
 \[
 \hat{\pi}_{\Cc_r}(\uv) = \Dc_r(\Ec_r(\uv))
 \]
 where $\Ec_r$ and $\Dc_r$ are the encoder and decoder of the compression algorithms respectively. There are two reasons for using this approximation for the projection operator: (i) All the state-of-the-art compression algorithms have computationally efficient decoders and encoders. Hence, the approximate projection will be fast as well. (ii) All state-of-the-art compression algorithms try to make sure that they are projecting each data point onto the closest codeword or a codeword that is in the vicinity of the closest codeword. Hence, this approximation is expected to be quite accurate.

Based on our discussion so far, our PGD algorithm proceeds according to the following iteration:
\begin{equation}\label{eq:approximate_proj}
\xv_t= \pi_{\Cc_r}(\xv_{t - 1} - \mu_{t} \gv_t),
\end{equation}
where $\xv_t$ is the estimate of $\xv_o$ at iteration $t$, $\gv_t$ is the gradient of $ \log\det(AX^2A^T)+{1\over \sigma_w^2}\yv^T(AX^2A^T)^{-1}\yv$ based on \eqref{eq:gradient}, and $\mu_t$ is the step-size at iteration $t$. There are two remaining ingredients in our algorithm: (i) the choice of the step-size $\mu_t$, and (ii) the choice of initialization. To set the step size we use the line search to find a value of $\mu_t$ that  makes $f(\Sigma(\pi_\Cc(\xv_{t - 1} - \mu \gv_t)))$ smaller than $f(\Sigma (\xv_{t - 1}))$. Additionally, we initialize the algorithm with $\xv_0$ being a constant vector. Based on all our discussions, the PGD-based algorithm operates as described in Algorithm~\ref{alg:1}. As discussed before, whenever possible we will use the exact choice of the projection function. Otherwise, we will use the approximation presented in \eqref{eq:approximate_proj}.


\begin{algorithm}[H]\label{alg:1}
\SetAlgoLined
\KwResult{$\xvh_T$}
 initialize $\xv_0\in\mathds{R}^n$\\
  $X_0=\diag(\xv_0)$,  \\
  $B_0=AX_0^2A^T$\;
 \For{$t=1,\ldots,T$}{
  \For{$i=1,\ldots,n$}{
   $s_{t,i} =x_{t-1,i} - \mu x_{t-1,i}(\av_i^TB_{t-1}^{-1}\av_i-(\av_i^TB_{t-1}^{-1} \yv)^2)$
  }
  $\xv_t=\pi_{\Cc_r}(\sv_t)$\\
  $X_t=\diag(\xv_t)$\\
  $B_t=AX_t^2A^T$
   }
 \caption{PGD-based recovery from under-determined measurements corrupted by multiplicative  noise}\label{alg1}
\end{algorithm}

\end{itemize}

\subsubsection{Piecewise constant functions}

To provide a concrete example for the discussions of the previous section, we mention a popular class of functions in imaging systems, i.e. the class of piecewise constant functions.  Refer to \cite{tibshirani2005sparsity} for applications of this model beyond imaging. Generalizations of our discussion to the class of piecewise polynomial functions are straightforward. However, to keep the notation simple, we focus on piecewise constant functions here. 

 Let $\Qc^{J,n}(x_{\min},x_{\max})$ denote the set of piecewise constant signals  in $\mathds{R}^n$ with the maximum of $J$ jumps and values bounded between $x_{\min}$ and $x_{\max}$, where $0<x_{\min} \leq x_{\max}$. To apply Alg.~\ref{alg:1}, we need to have a compression code for this class of signals.  In our simulations, we use a compression code that operates as follows.
\begin{itemize}
\item Encoder $\Ec_n$: Consider $\xv\in \Qc^{J,n}(x_{\min},x_{\max})$. Let $i_1,\ldots,i_j$ denote the location of the jumps in $\xv$. That is, $x_{i_k}\neq x_{i_k+1}$. (Note that $j\leq J$.) The encoder describes the number of jumps and the jump locations (using $(j+1)\lceil \log_2 J\rceil$ bits). Then, it describes the $b$-bit quantized versions of the  $j+1$ values  corresponding  to  the $j+1$ constant intervals (overall using $(j+1)(\lceil \log_2 ( x_{\max}-x_{\min}) \rceil +b)$ bits).   
\item Decoder $\Dc_n$: The decoder receives the number of jumps, locations of the jumps and the $b$-bit quantized values and reconstructs the signal accordingly. 
\end{itemize}
Let $\Cc_r$ denote the codebook corresponding to this compression code. Note that both the rate and the distortion depend on $J$ and $b$. For this class of functions (and also the class of piecewise polynomial functions), and for the codewords we have constructed, we can implement the exact projection function $ \pi_{\Cc_r}(\uv)$ efficiently using dynamic programming. 
Consider $\cv\in\Cc_r$, with $j$ jumps at $i_1<\ldots<i_j$ and values $a_1,\ldots,a_{j+1}\in\Xc_m$, where
\[
\Xc_b\triangleq \{[x]_b:\; x\in(x_{\min},x_{\max})\}.
\]
Let $i_0\triangleq 0$ and $i_{j+1}\triangleq n$. Then,
\begin{align}
\|\xv-\cv\|^2=\sum_{k=0}^{j}\sum_{i=i_k+1}^{i_{k+1}}(a_{k+1}-x_i)^2.\label{eq:cost-pw-const}
\end{align}
Fixing the jump locations, it is straightforward to find the minimizer of \eqref{eq:cost-pw-const}, as  
\[
\min_{(a_1,\ldots,a_{j+1})\in\Xc_b^{j+1}}\sum_{k=0}^{j}\sum_{i=i_k+1}^{i_{k+1}}(a_{k+1}-x_i)^2=\sum_{k=0}^{j}\min_{a_{k+1}\in\Xc_b}\sum_{i=i_k+1}^{i_{k+1}}(a_{k+1}-x_i)^2. 
\]
Therefore, the complexity of minimizing \eqref{eq:cost-pw-const} lies in finding jumps' locations. But, the expression of \eqref{eq:cost-pw-const} suggests that we can employ the Viterbi (dynamic programming) algorithm \cite{Viterbi:67} to find the  jumps' locations.  Consider a Trellis diagram with $J+1$ layers and $n$ nodes at each layer.  In this graph, layer $k$ corresponds to the $k$-th jump and node $i$ in layer $k$ correspond to having the $k$-th jump at $i$. Using this construction, node $i$ in layer $k$ only connects to nodes $i+1,\ldots,n$ at layer $k+1$. The weight assigned to the edge connection node $i_1$ in layer $k$ to node $i_2$ in layer $k+1$ is equal to 
\[
\min_{a\in\Xc_b} \sum_{i=i_1+1}^{i_2}(a-x_i)^2.
\]
Applying the Viterbi algorithm on the described Trellis diagram yields the solution of \eqref{eq:cost-pw-const}.

\subsection{Multilevel formulation}\label{ssec:multilevel}

In addition to the projected gradient descent algorithm proposed in the previous section, in some cases, such as the piecewise polynomial functions, we can reformulate our problem as a multilevel optimization problem. Namely, we split the optimization into a discrete component (optimization of the structure),
and a continuous (or near continuous) component (optimization of the values given the structure). To avoid introducing new notations, we describe this method on the concrete example of piecewise constant functions that we described in the previous section.  

Given integers $1 = d_1, d_2, \dotsc, d_{k + 1} = n + 1$ and positive real values $\theta_1, \dotsc, \theta_{k}$,
we may define the following piecewise-constant signal $\xv(\thetav, \dv)$:
\begin{equation*}
  \xv(\thetav, \dv)_i = \theta_l \quad \text{ if } d_{l} \leq i < d_{l + 1}.
\end{equation*}
We may thus view the recovery problem as:
\begin{equation*}
  \argmin_{\dv \in \NN^{k + 1}} \min_{\thetav \in \RR^{k}} f(\Sigma (\xv(\thetav, \dv))),\footnote{Note that in this formulation we have assumed we are doing the compression at its highest resolution, i.e. $b = \infty$ and $J= n$. It is straightforward to generalize this method to other cases as well. }
\end{equation*}
Now, note that the inner optimization problem, which we denote as:
\begin{equation*}
  h(\dv) = \min_{\thetav \in \RR^k} f(\Sigma(\xv(\thetav, \dv)))
\end{equation*}
can be solved using traditional gradient-descent or quasi-Newton type solvers, such as L-BFGS. Our simulation results show that the non-convexity of the cost does
not present any major challenges to these methods. Hence, suppose that once $\dv$ is given the calculation of the optimal solution is straightforward. We then attempt to optimize $h(\dv)$ directly as a function of $\dv$.
However, we note here that $h$ is a multivariate function of discrete quantities, and gradient information
(or similar) is not easily computable for $h$. Instead, we propose to use a gradient-free method to optimize $h(\dv)$.
We leverage recent advances in hyper-parameter optimization methods, and make use of \emph{optuna} \cite{optuna},
a general purpose gradient-free optimizer, to optimize $h$. Note that using alternating minimization on $\dv$ and $\thetav$, we can hope to obtain an approximate minimizer of the $f(\Sigma (\xv(\thetav, \dv)))$.



\section{Simulation results}\label{sec:sim}

As a preliminary investigation, we analyze the performance of the proposed methods on a simulated example.
To evaluate the performance of our algorithms, we consider the problem of recovering a piecewise-constant signal from its undersampled measurements.  We study the performance of the
following algorithms that are based on the two algorithms we discussed in the last section:
\begin{enumerate}
  \item A vanilla implementation of Algorithm \ref{alg:1}, where the step size $\mu$ is selected by line search at every iteration.
  We refer to this implementation as \emph{pgd}.
  \item We noticed a dependence of the optimization quality on the choice of initial value $x_0$. We consider an algorithm which performs
  Algorithm \ref{alg:1} for a set of initializations (namely all constant signals with a given magnitude), and selects the output among those which maximizes
  the likelihood of the data. We refer to this implementation as \emph{pgd + init}.
  \item We also consider the multi-level formulation described in section~\ref{ssec:multilevel}. We refer to this
  method as \emph{multilevel}.
  \item Additionally, we present a variant of this multi-level approach, where approximate break locations are first estimated using a PGD method, then neighborhoods
  of those breaks are searched using a gradient-free optimizer. We refer to this implementation as \emph{pgd + multilevel}.
\end{enumerate}
\begin{figure}
\begin{center}
\includegraphics[width=5cm]{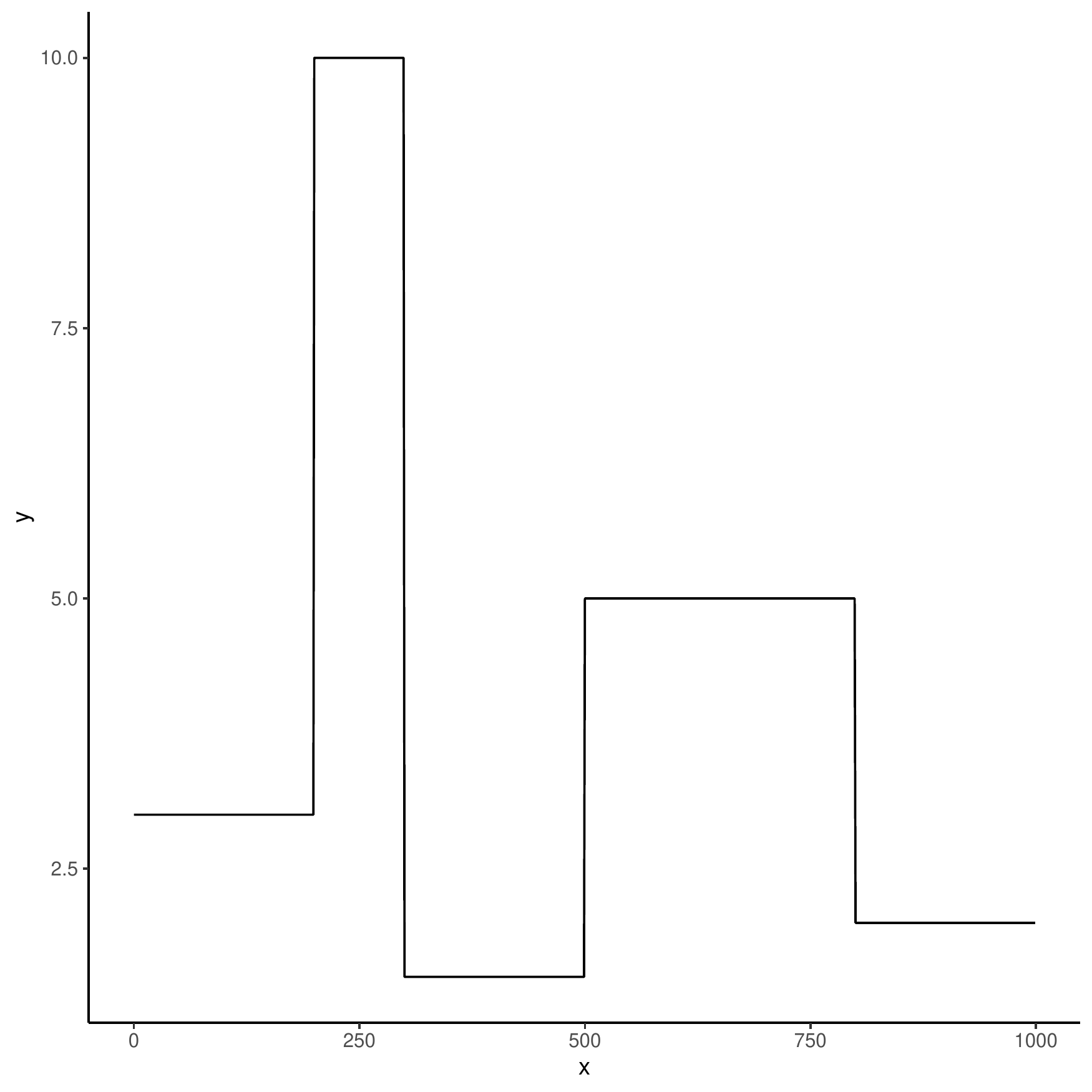}
\caption{The piecewise constant signal used in the simulations. \label{fig:truesignal_sim}}
\end{center}
\end{figure}

In our simulations, we consider a signal with $5$ pieces in dimension $n =1000$. This signal is shown in Figure \ref{fig:truesignal_sim}. We evaluate the performance of our algorithms for $m \in \{200, 250, 300, 350, 400, 450, 500 \}$ measurements. For each $m$, we generate 50 random measurement matrices and noise vectors and measure the reconstruction error of each of the four algorithms we mentioned above. To measure the amount of reconstruction error of vector $\hat{\xv}$ we use the PSNR in decibels defined as
\[
10 \log_{10} \left(\frac{\|\xv\|^2_{\infty}}{\|\hat{\xv} -\xv\|_2^2} \right)
\]

A few instances of reconstructions are shown in Figures \ref{fig:truesignal_vs_recon250} and \ref{fig:truesignal_vs_recon400}. A summary of our simulation results is presented in Figure~\ref{fig:method-performance}. We can draw the following conclusions based on these simulations:

\begin{figure}
\begin{center}
\includegraphics[width=10cm]{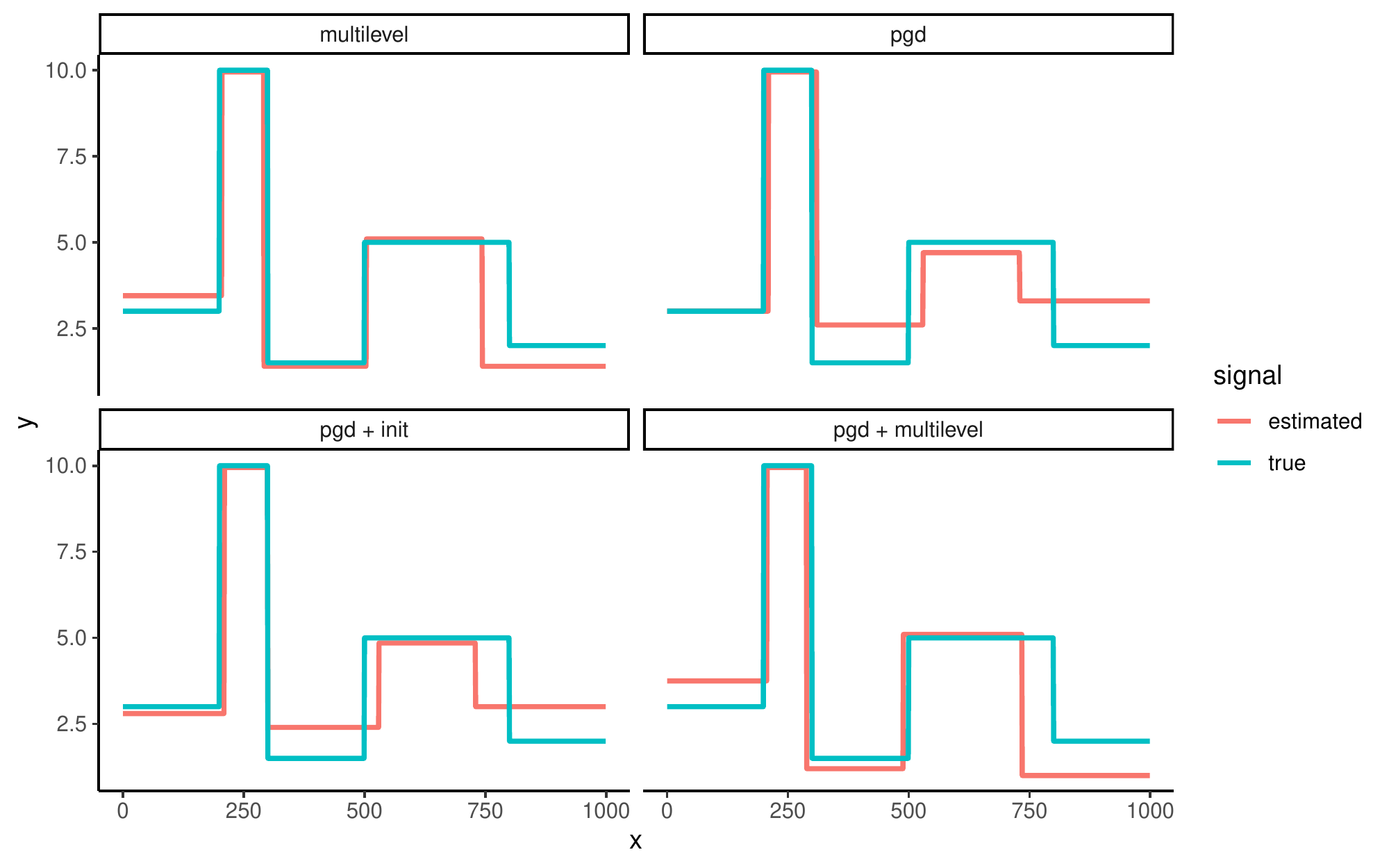}
\caption{Signal $\xv_o$ shown in Figure \ref{fig:truesignal_sim} and its reconstruction from its undersampled measurements with multiplicative noise using the four algorithms described in this section. The number of measurements is set to $m =250$.\label{fig:truesignal_vs_recon250}}
\end{center}
\end{figure}

\begin{figure}
\begin{center}
\includegraphics[width=10cm]{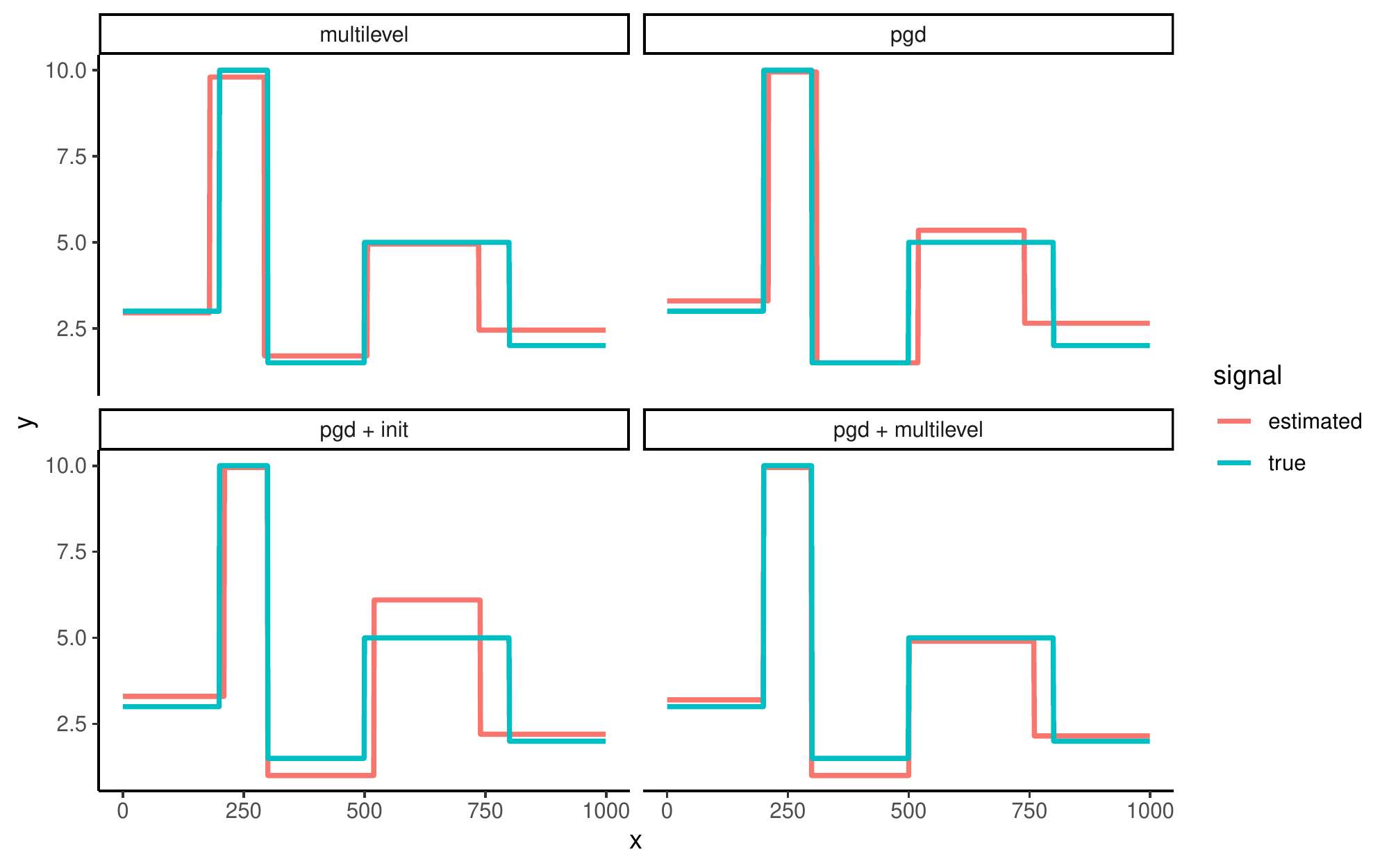}
\caption{Signal $\xv_o$ shown in Figure \ref{fig:truesignal_sim} and its reconstruction from its undersampled measurements with multiplicative noise using the four algorithms described in this section. The number of measurements is set to $m =400$. \label{fig:truesignal_vs_recon400}}
\end{center}
\end{figure}

\begin{enumerate}
\item As is clear from Figure \ref{fig:method-performance}, pgd+multilevel algorithm overall offers a better performance than the other algorithms. However, we should note that multilevel algorithms are often more computationally demanding (see Table~\ref{table:method-performance}), and face scalability issues to larger structures (e.g. larger number of constant pieces), or more complex structures (e.g. those that may be found in images). Finding algorithms that are more efficient than pgd-based methods and are still scalable for more sophisticated signals and structures is an important direction for our future research.

\item As can be seen in Figure ~\ref{fig:method-performance}(b), in general the true signal has higher log-likelihood than the solutions our algorithms are converging to. This is an indication of the fact that our algorithms have not been able to find the global minimizer of the negative log-likelihood. This again poses an open algorithmic problem. Can we find algorithms that achieve better solutions? Again this is an important direction for our future research.  

\end{enumerate}

\begin{figure}
  \centering
  \subcaptionbox{Statistical performance}{
    \includegraphics[width=0.45\textwidth]{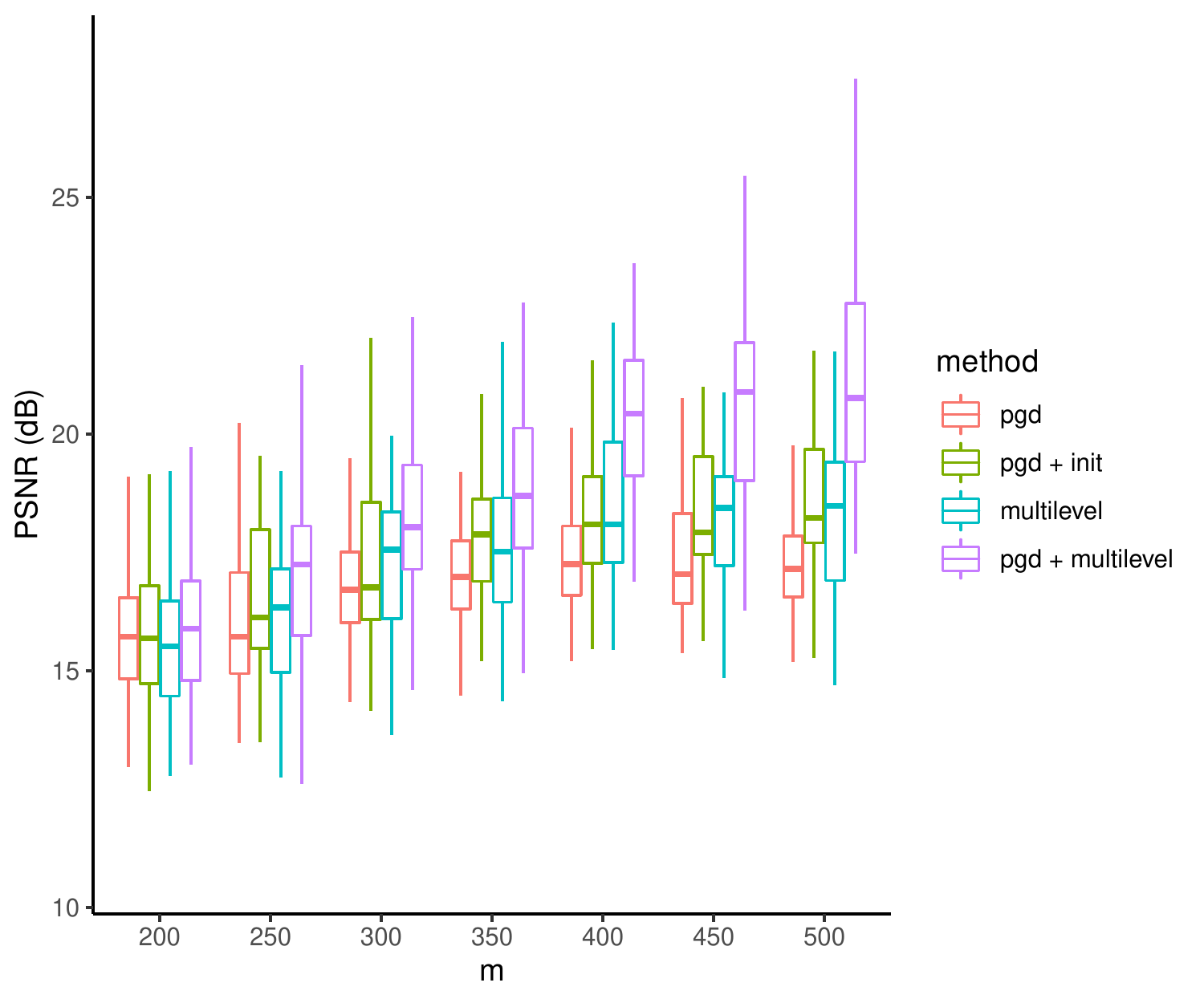}
  }
  \subcaptionbox{Optimization performance}{
    \includegraphics[width=0.45\textwidth]{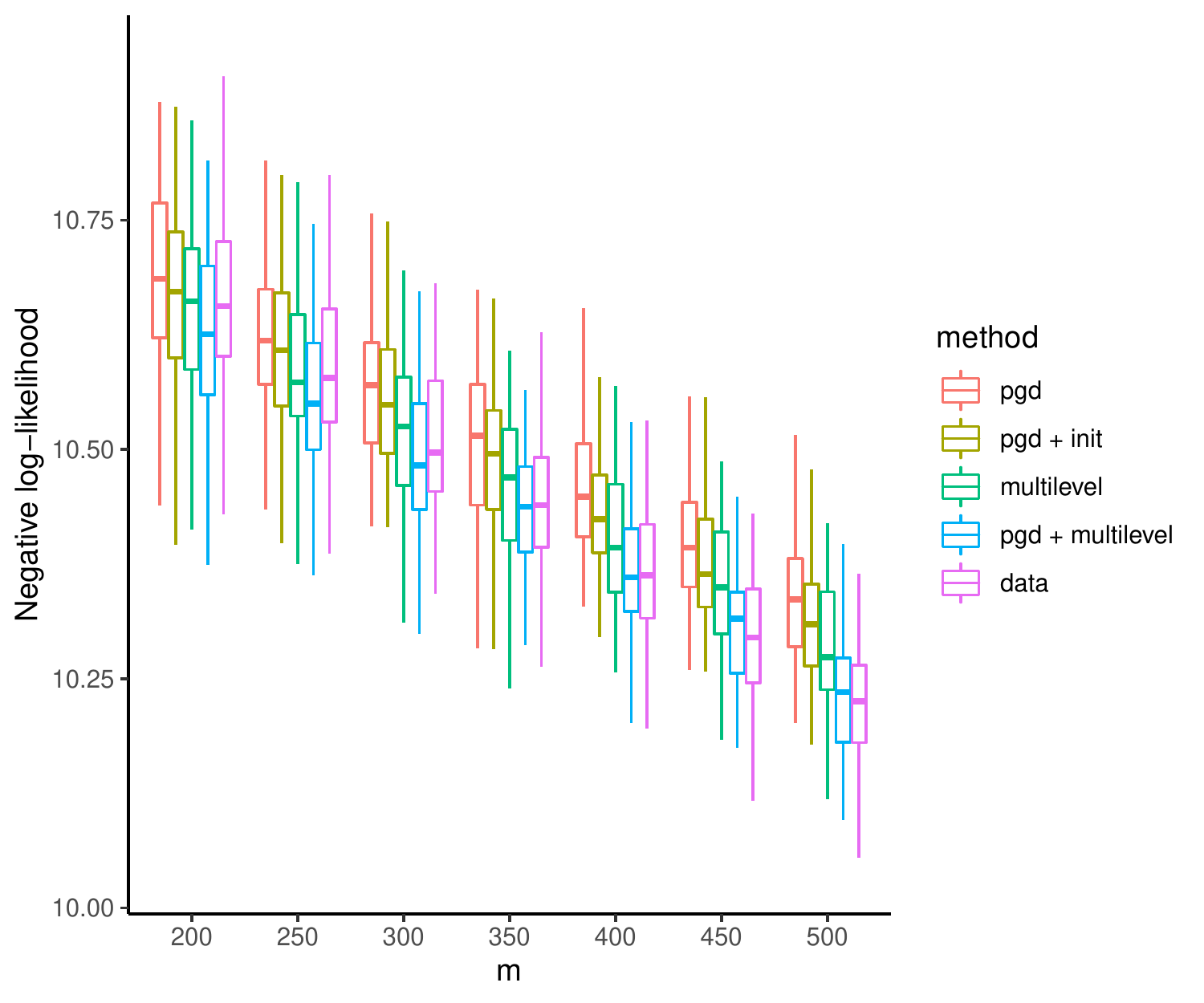}
  }

  \caption{Summary of our simulation results on a piecewise constant signal with $5$ pieces. Figure (a) exhibits the PSNR of the reconstruction of the four algorithms based on $50$ simulations as a function of the number of measurements $m$. In each simulation a fresh measurement matrix and random noise vector are drawn. Note that the bars represent the ninety percent confidence intervals. Figure (b) shows the negative log-likelihood of the solution of our algorithms. For comparison purposes we have also included the negative log-likelihood of $\xv_o$ the true signal. The label for the negative log-likelihood of the true signal is {\em data} in Figure (b). \label{fig:method-performance}}
\end{figure}

\begin{table}
  \centering
  \begin{tabular}{llrr}
    \toprule
    $m$ & Method & Time (s) & Evaluations \\
    \midrule
    200 & PGD          &  0.6  \makebox[\widthof{(00.0)}][r]{(0.2)}  & 89   \makebox[\widthof{(000)}][r]{(25)}  \\
    & PGD + init       &  8.8  \makebox[\widthof{(0.00)}][r]{(1.7)}  & 1325 \makebox[\widthof{(000)}][r]{(223)} \\
    & Multilevel       & 33.6  \makebox[\widthof{(0.00)}][r]{(1.2)}  & 4492 \makebox[\widthof{(000)}][r]{(55)}  \\
    & PGD + multilevel & 33.5  \makebox[\widthof{(0.00)}][r]{(1.6)}  & 4454 \makebox[\widthof{(000)}][r]{(77)}  \\
    400 & PGD          &  1.1  \makebox[\widthof{(0.00)}][r]{(0.3)}  & 65   \makebox[\widthof{(000)}][r]{(16)}  \\
    & PGD + init       & 20.6  \makebox[\widthof{(0.00)}][r]{(3.0)}  & 1173 \makebox[\widthof{(000)}][r]{(161)} \\
    & Multilevel       & 99.4  \makebox[\widthof{(0.00)}][r]{(5.7)}  & 4332 \makebox[\widthof{(000)}][r]{(35)}  \\
    & PGD + multilevel & 94.3  \makebox[\widthof{(0.00)}][r]{(15.7)} & 4302 \makebox[\widthof{(000)}][r]{(36)}  \\
    \bottomrule
  \end{tabular}

  \caption{Time and number of likelihood or gradient evaluations performed by the proposed methods: average and standard deviations are calculated based on $50$ Monte Carlo simulations. Note that the signal is piecewise constant with $5$ pieces in dimension $n=1000$. The number of measurements are mentioned in the first column. \label{table:method-performance}}
\end{table}

\section{Proofs}\label{sec:proof}

\subsection{Preliminaries}
Before stating the proofs, we present some lemmas, some new and some known, that are going to be used in the proofs on the main results.

\begin{lemma}\label{lem:singvalues}\cite{RudelsonVershinin2010}
Let the elements of an $m \times n$ (m<n) matrix $A$ be drawn independently from $\Nc(0,1)$. Then for any $t>0$ we have
\[
\P(\sqrt{n}-\sqrt{m}- t \leq \sigma_{\min} (A) \leq \sigma_{\max}(A) \leq \sqrt{n}+\sqrt{m}+ t) \geq 1-  2 {\rm e}^{-\frac{t^2}{2}}.
\] 
\end{lemma}

\begin{lemma}[Concentration of $\chi^2$ \cite{JalaliM:14-MEP-IT}]\label{lem:conc:chisq}
Let $Z_1, Z_2, \ldots, Z_n$ denote a sequence of independent $\Nc(0,1)$ random variables. Then, for any $t\in(0,1)$, we have
\[
\P (\sum_{i=1}^n Z_i^2 \leq m (1-t)) =  {\rm e}^{\frac{m}{2} (t + \log (1- t)) }. 
\]
Also, for any $t>0$,
\[
\P (\sum_{i=1}^n Z_i^2 \geq m (1+t)) =  {\rm e}^{-\frac{m}{2} (t -\log (1+ t)) }. 
\]
\end{lemma}

%

Define
\begin{equation}
\|X\|_{\psi_2} = \inf\{t>0 : \mathbb{E} (\exp (X^2/t^2)) \leq 2\}. 
\end{equation}

\begin{theorem}[Hanson-Wright inequality] \label{thm:HW-ineq}
Let $\Xv= (X_1,...,X_n)$ be a random vector with independent components with $\E[X_i]=0$ and $\|X_i\|_{\Psi_2}\leq K$. Let A be an $n\times n$ matrix. Then, for $t>0$, 
\[
\P\Big(|\Xv^T A\Xv-\E[\Xv^T A\Xv]|>t \Big)\leq 2\exp\Big(-c\min({t^2\over K^4\|A\|_{\rm HS}^2 },{t \over K^2\|A\| })\Big).
\]
\end{theorem}

\begin{lemma}\label{lemma:1}
Let $\lambda_{\max}$ denote the maximum eigenvalue of $\Sigma_o^{-{1\over 2}}\Delta \Sigma\Sigma_o^{-{1\over 2}}$. Then,
\begin{align}
{1\over 2(1+\lambda_{\max})^2}{\rm Tr}(\Sigma_o^{-1}\Delta \Sigma\Sigma_o^{-1}\Delta \Sigma)\leq  \bar{f}(\Sigma)-\bar{f}(\Sigma_o)&\leq {1\over 2}{\rm Tr}(\Sigma_o^{-1}\Delta \Sigma\Sigma_o^{-1}\Delta \Sigma),
\end{align}
where  $\Delta \Sigma=\Sigma-\Sigma_o$. 
\end{lemma}

\begin{proof}[Proof of Lemma \ref{lemma:1}]
Note that 
\begin{align}
\bar{f}(\Sigma)-\bar{f}(\Sigma_o)&=-\log\det \Sigma +{\rm Tr}(\Sigma AX_o^2A^T) +\log\det \Sigma_o -{\rm Tr}(\Sigma_o AX_o^2A^T).\label{eq:step1}
\end{align}
But,
\begin{align}
\log\det \Sigma =\log\det(\Sigma_o+\Delta \Sigma)=\log\det \Sigma_o^{1\over 2}(I+\Sigma_o^{-{1\over 2}}\Delta \Sigma\Sigma_o^{-{1\over 2}})\Sigma_o^{1\over 2}.
\end{align}
Therefore,
\begin{align}
\log\det \Sigma-\log\det \Sigma_o &=\log\det(\Sigma_o+\Delta \Sigma)=\log\det (I+\Sigma_o^{-{1\over 2}}\Delta \Sigma\Sigma_o^{-{1\over 2}})\nonumber\\
&=\sum_{i=1}^n\log(1+\lambda_i(\Sigma_o^{-{1\over 2}}\Delta \Sigma\Sigma_o^{-{1\over 2}})),
\end{align}
where $\lambda_i=\lambda_i(\Sigma_o^{-{1\over 2}}\Delta \Sigma\Sigma_o^{-{1\over 2}})$, $i=1,\ldots,n$.  Moreover, by the mean value theorem, for $x>0$, $\log(1+x)= x-{1\over 2(1+a)^2}x^2$, for some $a\in(0,x)$.  Therefore, since $\lambda_i\geq 0$, for all $i$, we have  
\begin{align}
\log\det \Sigma-\log\det \Sigma_o &\leq \sum_{i=1}^n\lambda_i -{1\over 2(1+\lambda_{\max})^2} \sum_{i=1}^n \lambda_i^2\nonumber\\
&= {\rm Tr}(\Sigma_o^{-1}\Delta \Sigma)-{1\over 2(1+\lambda_{\max})^2}{\rm Tr}(\Sigma_o^{-1}\Delta \Sigma\Sigma_o^{-1}\Delta \Sigma).
\end{align}
Combining this bound with \eqref{eq:step1}, the desired result follows
\begin{align}\label{eq:upper_bound}
\bar{f}(\Sigma)-\bar{f}(\Sigma_o)&\geq  -{\rm Tr}(\Sigma_o^{-1}\Delta \Sigma)+{1\over 2(1+\lambda_{\max})^2}{\rm Tr}(\Sigma_o^{-1}\Delta \Sigma\Sigma_o^{-1}\Delta \Sigma) +{\rm Tr}(\Sigma AX_o^2A^T)  -{\rm Tr}(\Sigma_o AX_o^2A^T)\nonumber\\
&={1\over 2(1+\lambda_{\max})^2}{\rm Tr}(\Sigma_o^{-1}\Delta \Sigma\Sigma_o^{-1}\Delta \Sigma).
\end{align}
Note that to obtain an upper bound for $\bar{f}(\Sigma)-\bar{f}(\Sigma_o)$ we should replace $\lambda_{\max}$ with zero in \eqref{eq:upper_bound}.

\end{proof}

\begin{lemma}\label{lemma:concent-z}
Given $\Sigma=(AX^2A^T)^{-1}$, let $\delta f({\Sigma})={f}({\Sigma})-\bar{f}({\Sigma})$. Then, for $t>0$, 
\[
\P(|\delta f(\Sigma)|\geq t|A)\leq  2\exp\Big(-{ct\min(1,{t\over 2m  }) \over 2 x_{\max}^4\|A^T\Sigma A\|^2}\Big).
\]
Also, $\|A^T\Sigma A\|^2\leq {\lambda^2_{\max}(A^TA)\over \lambda^2_{\min}(AA^T) x^4_{\min}}$. 
\end{lemma}

\begin{proof}[Proof of Lemma \ref{lemma:concent-z}]
By definition, 
\[
\delta f({\Sigma})={f}({\Sigma})-\bar{f}({\Sigma})={1\over \sigma_w^2}\wv^TX_oA^T\Sigma AX_o\wv-{\rm Tr}(\Sigma AX_o^2A^T)
\]
Define matrix $B\in\mathds{R}^{n\times n}$ as  $B=X_oA^T\Sigma AX_o$. Then, be the Hanson-Wright inequality (Theorem \ref{thm:HW-ineq}), we have 
\begin{align}
\P(|{1\over \sigma_w^2}\wv^TB\wv- {\rm Tr}(\Sigma AX_o^2A^T)|>t)\leq 2\exp\Big(-c\min({t^2\over 4\|B\|_{\rm HS}^2 },{t \over 2\|B\|^2 })\Big).
\end{align}
But
\begin{align}
\|B\|_{\rm HS}^2={\rm Tr}(B^2)=\sum_{i=1}^m\lambda_i^2(B)\leq m\lambda^2_{\max}(B)=m\|B\|^2.
\end{align}
On the other hand, $\| B\|=\|X_oA^T\Sigma AX_o\|\leq x^2_{\max}\|A^T\Sigma A\|$. Moreover,
\begin{align}
\|A^T\Sigma A\|^2 =\max_{\uv\in\mathds{R}^n}{ \uv^T A^T \Sigma A A^T\Sigma A\uv \over \|\uv\|^2}\leq \lambda_{\max}(A^TA)\lambda_{\max}(AA^T)\lambda_{\max}^2(\Sigma)
\end{align}
But $\Sigma=(AX^2A^T)^{-1}$ and $X={\rm diag}(\xv)$.  Therefore, $\lambda_{\max}(\Sigma)=(\lambda_{\min}(AX^2A^T))^{-1}\leq (\lambda_{\min}(AA^T)x^2_{\min})^{-1}$ and
\begin{align}
\|A^T\Sigma A\|^2\leq {\lambda_{\max}(AA^T)\lambda_{\max}(A^TA)\over \lambda^2_{\min}(AA^T) x^4_{\min}}.
\end{align}

\end{proof}

\begin{lemma}\label{lemma:vector}
Consider $m\times m$ matrix defined as $\Delta \Sigma=\Sigma-\Sigma_o$, where $\Sigma=(AX^2A^T)^{-1}$ and $\Sigma=(AX_o^2A^T)^{-1}$.  Then,
\[
  {x^4_{\min} \lambda^3_{\min}(AA^T) \over x^8_{\max} \lambda^4_{\max}(AA^T)  }\|A(X^2-X_o^2)\|^2_{\rm HS}
 \leq {\rm Tr}(\Sigma_o^{-1}\Delta \Sigma\Sigma_o^{-1}\Delta \Sigma)\leq {x^4_{\max} \lambda^3_{\max}(AA^T) \over x_{\min}^8 \lambda^4_{\min}(AA^T)  }\|A(X^2-X_o^2)\|^2_{\rm HS} 
\]
\end{lemma}

\begin{proof}[Proof of Lemma \ref{lemma:vector}]
Using the vectorizing equality, 
\begin{align}
{\rm Tr}(\Sigma_o^{-1}\Delta \Sigma\Sigma_o^{-1}\Delta \Sigma)&\stackrel{\rm (a)}{=}({\rm Vec}( \Delta \Sigma))^T (\Sigma_o^{-1} \otimes \Sigma_o^{-1}) {\rm Vec}( \Delta \Sigma)\nonumber\\
&\stackrel{\rm (b)}{\geq}  \| {\rm Vec}( \Delta \Sigma)\|_2^2 \lambda_{\min}(\Sigma_o^{-1} \otimes \Sigma_o^{-1}))\nonumber\\
&\stackrel{\rm (c)}{=}   \|\Delta \Sigma \|^2_{\rm HS} \lambda^2_{\min}(\Sigma_o^{-1})\nonumber\\
&=     \|\Delta \Sigma \|^2_{\rm HS} \lambda^2_{\min}(AX_o^2A^T),\label{eq:lemma-vec-s1}
\end{align}
where  $\otimes$ denotes the Kronecker product operation. Steps (a), (b) and (c) follow from the results on vectorization and Kronecker product. (Refer for instance to  Chapter 13 of \cite{laub2005matrix}.) But
\begin{align}
\lambda_{\min}(AX_o^2A^T)&=\min_{\uv\neq 0}{\|X_oA^T\uv\|^2 \over \|A^T\uv\|^2}{\|A^T\uv\|^2\over \|\uv\|^2}\geq x^2_{\min}\lambda_{\min}(AA^T),\label{eq:lemma-vec-s2}
\end{align}
which yields the desired result. 

On the other hand, $\Sigma_o-\Sigma= \Sigma(\Sigma^{-1}-\Sigma_o^{-1})\Sigma_o= \Sigma A(X^2-X_o^2)A^T\Sigma_o$. Moreover, $\|BC\|_{\rm HS}\geq \sigma_{\min}(B)\|\|C\|_{\rm HS}$ and $\|BC\|_{\rm HS}\geq  \sigma_{\min}(C)\|B\|_{\rm HS}$. Therefore, 
\begin{align}
\|\Sigma_o-\Sigma\|_{\rm HS}&\geq \lambda_{\min}(\Sigma)\sigma_{\min}(A^T\Sigma_o) \|A(X^2-X_o^2)\|_{\rm HS}\nonumber\\
&\geq {\lambda^{1\over 2}_{\min}(AA^T) \over \lambda_{\max}(AX^2A^T) \lambda_{\max}(AX^2_oA^T) }\|A(X^2-X_o^2)\|_{\rm HS}\nonumber\\
&\geq {\lambda^{1\over 2}_{\min}(AA^T) \over x_{\max}^4 \lambda^2_{\max}(AA^T)  }\|A(X^2-X_o^2)\|_{\rm HS}.\label{eq:lemma-vec-s3}
\end{align}
Combining \eqref{eq:lemma-vec-s1},  \eqref{eq:lemma-vec-s2} and \eqref{eq:lemma-vec-s3} yields the desired lower bound.  Similarly, to obtain an upper bound for ${\rm Tr}(\Sigma_o^{-1}\Delta \Sigma\Sigma_o^{-1}\Delta \Sigma)$, note that
\begin{align}
{\rm Tr}(\Sigma_o^{-1}\Delta \Sigma\Sigma_o^{-1}\Delta \Sigma)& {=}({\rm Vec}( \Delta \Sigma))^T (\Sigma_o^{-1} \otimes \Sigma_o^{-1}) {\rm Vec}( \Delta \Sigma)\nonumber\\
&{\leq}  \| {\rm Vec}( \Delta \Sigma)\|_2^2 \lambda_{\max}(\Sigma_o^{-1} \otimes \Sigma_o^{-1}))\nonumber\\
&=   \|\Delta \Sigma \|^2_{\rm HS} \lambda^2_{\max}(\Sigma_o^{-1})\nonumber\\
&=     \|\Delta \Sigma \|^2_{\rm HS} \lambda^2_{\max}(AX_o^2A^T) \nonumber\\
&\leq   \|\Delta \Sigma \|^2_{\rm HS} x^2_{\max} \lambda_{\max} (AA^T). 
\end{align}
Furthermore, using similar techniques as those  used in deriving \eqref{eq:lemma-vec-s3}, we have 
\begin{align}
\|\Sigma_o-\Sigma\|_{\rm HS}&\leq \lambda_{\max}(\Sigma)\sigma_{\max}(A^T\Sigma_o) \|A(X^2-X_o^2)\|_{\rm HS}\nonumber\\
&\leq {\lambda^{1\over 2}_{\max}(AA^T) \over \lambda_{\min}(AX^2A^T) \lambda_{\min}(AX^2_oA^T) }\|A(X^2-X_o^2)\|_{\rm HS}\nonumber\\
&\leq {\lambda^{1\over 2}_{\max}(AA^T) \over x_{\min}^4 \lambda^2_{\min}(AA^T)  }\|A(X^2-X_o^2)\|_{\rm HS}.\label{eq:lemma-vec-s3}
\end{align}

\end{proof}

%

\begin{lemma}\label{lem:lowedbound:AD}
Let the elements of $m\times n$ matrix $A$ be drawn  i.i.d.~ $\Nc(0,1)$. For any given $\bf{d} \in \mathbb{R}^n$, define $D = \diag (\bf{d})$. Then, for any $t\in(0,1)$,
\[
\P(\|A D\|^2_{HS} \leq m (1-t) \| {\bf d}\|_2^2, \;\forall \; {\bf d}) \leq n \exp(\frac{m}{2} (t + \log (1- t))), 
\]
and, for any $t>0$,
\[
\P(\|A D\|^2_{HS} \geq m (1+t) \| {\bf d}\|_2^2,  \;\forall \; {\bf d}) \leq n \exp(-\frac{m}{2} (t - \log (1+ t))).
\]
 \end{lemma}
\begin{proof}
Let $A=[\av_1;\ldots,\av_n]$ with $\av_i\in\mathds{R}^m$. Then, by definition, 
\begin{equation}\label{eq:lemma}
\|A D\|^2_{HS} = \sum_{i=1}^n d_i^2 \|\av_i\|_2^2. 
\end{equation}
Define event $\mathcal{E}$ as the event where $\min_i \|\av_i\|_2^2 > m(1- t)$. Using Lemma \ref{lem:conc:chisq} and applying the union bound, it follows that $\P(\mathcal{E}^c)  \geq n  {\rm e}^{\frac{m}{2} (t + \log (1- t)) }$. It is clear that under $\mathcal{E}$, we have
\[
\|A D\|^2_{HS} = \sum_{i=1}^n d_i^2 \|A_i\|_2^2 \geq m(1-t) \| {\bf d}\|_2^2. 
\]
\end{proof}

\subsection{Proof of Theorem \ref{thm:A-At-invert}}\label{proof_thm:A-At-invert}
Using the matrix inversion lemma, we have 
\begin{align}
(I_n+ { \sigma_w^2 \over \sigma_z^2}XA^TAX)^{-1}= I_n-{\sigma_w^2\over \sigma_z^2}XA^T\Big(I_n+{\sigma_w^2\over \sigma_z^2}AX^2A^T\Big)^{-1}AX.
\end{align}
Inserting this in \eqref{eq:def-ell-X}, it follows that 
\begin{align}
2\ell(X)&=-\log\det({1\over \sigma_w^2}I_n+ {1\over \sigma_z^2}XA^TAX)+{\sigma_w^2\over \sigma_z^4}\yv^TAX^2A^T\yv\nonumber\\
&\;\;\;\;- {\sigma_w^4\over \sigma_z^6}\yv^TAX^2A^T\Big(I_n+{\sigma_w^2\over \sigma_z^2}AX^2A^T\Big)^{-1}AX^2A^T\yv.
\end{align}
But,
\begin{align}
& {\sigma_w^4\over \sigma_z^6}\yv^TAX^2A^T\Big(I_n+{\sigma_w^2\over \sigma_z^2}AX^2A^T\Big)^{-1}AX^2A^T\yv\nonumber\\
& = {\sigma_w^2\over \sigma_z^4}\yv^T(I_n+{\sigma_w^2\over \sigma_z^2}AX^2A^T-I_n)\Big(I_n+{\sigma_w^2\over \sigma_z^2}AX^2A^T\Big)^{-1}AX^2A^T\yv\nonumber\\
  &={\sigma_w^2\over \sigma_z^4}\yv^TAX^2A^T\yv-{\sigma_w^2\over \sigma_z^4}\yv^T\Big(I_n+{\sigma_w^2\over \sigma_z^2}AX^2A^T\Big)^{-1}AX^2A^T\yv.
\end{align}
Therefore, cancelling the common terms, we have
\begin{align}
2\ell(X)&=-\log\det({1\over \sigma_w^2}I_n+ {1\over \sigma_z^2}XA^TAX)+{\sigma_w^2\over \sigma_z^4}\yv^T\Big(I_n+{\sigma_w^2\over \sigma_z^2}AX^2A^T\Big)^{-1}AX^2A^T\yv\nonumber\\
&=-\log\det({1\over \sigma_w^2}I_n+ {1\over \sigma_z^2}XA^TAX)+{1\over \sigma_z^2}\yv^T\Big(I_n+{\sigma_w^2\over \sigma_z^2}AX^2A^T\Big)^{-1}\Big({\sigma_w^2\over \sigma_z^2}AX^2A^T+I_n-I_n\Big)\yv\nonumber\\
&=-\log\det({1\over \sigma_w^2}I_n+ {1\over \sigma_z^2}XA^TAX)+ {\yv^T\yv\over \sigma_z^2}-{1\over \sigma_z^2}\yv^T\Big(I_n+{\sigma_w^2\over \sigma_z^2}AX^2A^T\Big)^{-1}\yv.
\end{align}
Assuming that $\sigma_z$ is converging to zero, ignoring the terms not depending on $X$, we have 
\begin{align}
-2\ell(X)&=\log\det({1\over \sigma_w^2}I_n+ {1\over \sigma_z^2}XA^TAX)+{1\over \sigma_w^2}\yv^T(AX^2A^T)^{-1}\yv.
\end{align}
Let $\lambda_1,\ldots,\lambda_m$ denote the non-zero eigenvalues of $XA^TAX$. Then, ${1\over \sigma_w^2}I_n+ {1\over \sigma_z^2}XA^TAX$ has $n-m$ eigenvalues equal to ${1\over \sigma_w^2}$ and the rest of its eigenvalues are ${1\over \sigma_w^2}+{1\over \sigma_z^2}\lambda_i$, $i=1,\ldots,m$.
Therefore,
\begin{align}
\log\det({1\over \sigma_w^2}I_n+ {1\over \sigma_z^2}XA^TAX)
&=-2(n-m)\log  \sigma_w+\sum_{i=1}^m\log({1\over \sigma_w^2}+{1\over \sigma_z^2}\lambda_i)\nonumber\\
&=-2(n-m)\log  \sigma_w-2m\log\sigma_z+\sum\log\lambda_i+O(\sigma_z^2)\nonumber\\
&=\sum_{i=1}^m\log\lambda_i-2(n-m)\log  \sigma_w-2m\log\sigma_z+O(\sigma_z^2).
\end{align}
Note that if $XA^TAX\uv=\lambda_i\uv$, then $AX^2A^T(AX\uv)=\lambda_i(AX\uv)$. Therefore, $\lambda_i$, $i=1,\ldots,m$, are the eigenvalues of $AX^2A^T$ as well. Hence, in summary, again by ignoring the terms that do not depend on $X$, by a slight abuse of  notation, as $\sigma_z\to 0$, we have
\begin{align}
-2\ell(X)&=\log\det(AX^2A^T)+{1\over \sigma_w^2}\yv^T(AX^2A^T)^{-1}\yv.
\end{align}

\subsection{Proof of Theorem \ref{thm:main-result}}

Since $\yv=AX_o\wv$, function $f$ can be written as $f(\Sigma)=-\log\det \Sigma +{1\over \sigma_w^2}{\rm Tr}(\Sigma AX_o\wv\wv^TX_oA^T)$. Taking the expected value of $f$ with respect to the noise $\wv$, we derive
\begin{align}
\bar{f}(\Sigma)=-\log\det \Sigma +{\rm Tr}(\Sigma AX_o^2A^T).\label{eq:min-f-bar}
\end{align}
As a function of $\Sigma$, $\bar{f}$ is a convex function that achieves its minimum at $\Sigma^{-1}= AX_o^2A^T$ or at $X$ satisfying $AX^2A^T= AX_o^2A^T$. Define
 \[
 \Sigma_o\triangleq (AX_o^2A^T)^{-1}.
 \]

 Given $\hat{\xv}_o$ (the minimizer of $f$),  let $\hat{\Sigma}_o=\Sigma(\hat{\xv}_o)$.  Moreover, define $\xvt_o$ as the closest reconstruction signal in $\Cc$ to $\xv_o$, i.e.,
 \[
\xvt_o=\argmin_{\xv\in\Cc}\|\xv_o-\xv\|.
 \]
 Let $ \tilde{\Sigma}_o\triangleq (A\Xt_o^2A^T)^{-1}$, where $\Xt_o={\rm diag}({\xvt_o})$. Since $\xvh_o$ is the minimizer of \eqref{eq:min-f-main}, 
\begin{align}
f(\hat{\Sigma}_o)\leq f(\tilde{\Sigma}_o). \label{eq:thm1-step1}
\end{align}

Define $\Delta \Sigma$ as $\hat{\Sigma}_o-\Sigma_o$ and let $\lambda_{m}$ denote the maximum eigenvalue of $\Sigma_o^{-{1\over 2}}\Delta \Sigma\Sigma_o^{-{1\over 2}}$. Then, as shown in Lemma \ref{lemma:1}, $\bar{f}(\hat{\Sigma}_o)-\bar{f}(\Sigma_o)$ can be lower bounded as follows:
\begin{align}
\bar{f}(\hat{\Sigma}_o)-\bar{f}(\Sigma_o)&\geq {1\over 2(1+\lambda_{m})^2}{\rm Tr}(\Sigma_o^{-1}\Delta \Sigma\Sigma_o^{-1}\Delta \Sigma),\label{eq:thm1-step2}
\end{align}
where  $\Delta \Sigma=\hat{\Sigma}_o-\Sigma_o$.  Let $\delta f(\hat{\Sigma}_o)\triangleq{f}(\hat{\Sigma}_o)-\bar{f}(\hat{\Sigma}_o)$ and $\delta f(\tilde{\Sigma}_o)\triangleq{f}(\tilde{\Sigma}_o)-\bar{f}(\tilde{\Sigma}_o)$. Then,
\begin{align}
\bar{f}(\hat{\Sigma}_o)-\bar{f}(\Sigma_o)&={f}(\hat{\Sigma}_o)-\delta f(\hat{\Sigma}_o)-(\bar{f}(\Sigma_o)-{f}(\tilde{\Sigma}_o)+{f}(\tilde{\Sigma}_o)-\bar{f}(\tilde{\Sigma}_o)+\bar{f}(\tilde{\Sigma}_o))\nonumber\\
&={f}(\hat{\Sigma}_o)-{f}(\tilde{\Sigma}_o)-\delta f(\hat{\Sigma}_o)+\delta f(\tilde{\Sigma}_o)-\bar{f}(\Sigma_o)+\bar{f}(\tilde{\Sigma}_o).
\end{align}
Therefore, combining \eqref{eq:thm1-step1} and \eqref{eq:thm1-step2}, it follows that
\[
 {1\over 2(1+\lambda_{m})^2}{\rm Tr}(\Sigma_o^{-1}\Delta \Sigma\Sigma_o^{-1}\Delta \Sigma)\leq |\delta f(\tilde{\Sigma}_o)| +|\delta f(\hat{\Sigma}_o)|+|\bar{f}(\Sigma_o)-\bar{f}(\tilde{\Sigma}_o)|.
\]
Also, applying Lemma \ref{lemma:vector}, it follows that
\[
 {1\over 2(1+\lambda_{m})^2} {x^4_{\min} \lambda^3_{\min}(AA^T) \over x^8_{\max} \lambda^4_{\max}(AA^T)  }\|A(\hat{X}_o^2-X_o^2)\|^2_{\rm HS}\leq |\delta f(\tilde{\Sigma}_o)| +|\delta f(\hat{\Sigma}_o)|+|\bar{f}(\Sigma_o)-\bar{f}(\tilde{\Sigma}_o)|.
\]
Given $t_1\in(0,1)$, $t_2>0$ and $t_3>0$,  define events $\Ec_1$,  $\Ec_2$, $\Ec_3$ and $\Ec_4$ as
\[
\Ec_1= \{\|A(\hat{X}_o^2-X_o^2)\|^2_{\rm HS}\geq m(1-t_1) \|\xv^2_o-\xvh^2_o\|^2\},
\]
\[
\Ec_2=\{ |\delta f(\tilde{\Sigma}_o)| \leq t_2\},  \;\;\;\Ec_3=\{ |\delta f(\hat{\Sigma}_o)|\leq t_2\},
\]
\[
\Ec_4=\{\sqrt{n}-2\sqrt{m}\leq \sigma_{\min} (A) \leq \sigma_{\max}(A) \leq \sqrt{n}+2\sqrt{m}\},
\]
and
\[
\Ec_5=\{\|A(\tilde{X}_o^2-X^2_o)\|^2_{\rm HS} \leq m(1+t_3)\|\xv_o^2-{\xvt^2_o}\|^2\},
\]
respectively. 
Conditioned on $\Ec_1\cap\Ec_2\cap\Ec_3\cap \Ec_4$, we have
\begin{equation}\label{eq:upperbound:almostlast}
 {m \over 2(1+\lambda_{m})^2} {x^4_{\min} (\sqrt{n} -2 \sqrt{m})^6 \over x^8_{\max} (\sqrt{n} +2 \sqrt{m})^8  } (1-t_1)\| \xv^2- \xvh_o^2\|_2^2\leq 2t_2+|\bar{f}(\Sigma_o)-\bar{f}(\tilde{\Sigma}_o)|.
\end{equation} 
But, from Lemma \ref{lemma:1}, we have
\begin{align}
|\bar{f}(\Sigma_o)-\bar{f}(\tilde{\Sigma}_o)|&\leq {1\over 2}{\rm Tr}(\Sigma_o^{-1}(\tilde{\Sigma}_o-\Sigma_o)\Sigma_o^{-1}(\tilde{\Sigma}_o-\Sigma_o))\nonumber\\
&\leq {x^4_{\max} \lambda^3_{\max}(AA^T) \over 2 {x_{\min}^8} \lambda^4_{\min}(AA^T)  }\|A(\tilde{X}^2_o-X^2_o)\|^2_{\rm HS},\label{eq:step1-main}
\end{align}
where the second inequality  follows from Lemma \ref{lemma:vector}. Then, conditioned on $\Ec_4\cap\Ec_5$, it follows that 
\begin{align}
|\bar{f}(\Sigma_o)-\bar{f}(\tilde{\Sigma}_o)|&\leq  {x^4_{\max}(\sqrt{n}+2\sqrt{m})^6  \over 2 {x_{\min}^8}  (\sqrt{n}-2\sqrt{m})^8  }m(1+t_3)\|\xv_o^2-{\xvt^2_o}\|^2.\label{eq:step2-main}
\end{align}
Define 
\[
c_1 = {x^4_{\min} (1 -2 \sqrt{m/n})^6 \over x^8_{\max} (1 +2 \sqrt{m/n})^8  },
\]
and
\[
c_2 =  {x^4_{\max}(1+2\sqrt{m/n})^6  \over{x_{\min}^8}  (1-2\sqrt{m/n})^8  }.
\]
Using these definitions and combining \eqref{eq:step1-main} and \eqref{eq:step2-main}, it follows that 
\begin{align}
 {m \over 2(1+\lambda_{m})^2} {c_1\over n} (1-t_1)\| \xv_o^2- \xvh^2_o\|_2^2 \leq   2t_2+{c_2\over 2n}   m(1+t_3)\|\xv^2_o-{\xvt^2_o}\|^2.\label{eq:step3-main}
\end{align}
Recall that $\lambda_{m}$ is defined as the maximum eigenvalue of $\Sigma_o^{-{1\over 2}}\Delta \Sigma\Sigma_o^{-{1\over 2}}$. On the other hand, $\lambda_m=\|\Sigma_o^{-{1\over 2}}\Delta \Sigma\Sigma_o^{-{1\over 2}}\|\leq \|\Delta \Sigma\| \|\Sigma_o^{-{1\over 2}}\|^2 = \|\Delta \Sigma\| \|\Sigma_o^{-1}\|$. But, $ \|\Sigma_o^{-1}\|=\|AX_o^2A^T\|\leq x^2_{\max}\lambda_{\max}(AA^T)$. Similarly, $\|\Delta \Sigma\|=\|\hat{\Sigma}_o-{\Sigma}_o\| \leq \|\hat{\Sigma}_o\|-\|{\Sigma}_o\| \leq {1\over \|AX_o^2A^T\|}+{1\over \|A\hat{X}_o^2A^T\|}\leq {2\over x^2_{\min}\lambda_{\min}(AA^T)}$. So overall, $\lambda_m\leq {2x^2_{\max}\lambda_{\max}(AA^T)\over x^2_{\min}\lambda_{\min}(AA^T)}$, and conditioned on $\Ec_4$, we have
\begin{align}
\lambda_m\leq {2x^2_{\max}(1+2\sqrt{m/n})^2\over x^2_{\min}(1-2\sqrt{m/n})^2}.\label{eq:def-lambda-m}
\end{align}

To finish the proof we bound the probability of $\cap_{i=1}^5 \Ec_i$.  Lemma \ref{lem:lowedbound:AD} can be used to bound $\P(\Ec_1^c)+\P(\Ec_5^c)$ as
\[
\P(\Ec_1^c)\leq  n \exp(\frac{m}{2} (t_1 + \log (1- t_1))),
\]
and
\[
 \P(\Ec_5^c)\leq n\exp(-\frac{m}{2} (t_3 - \log (1+ t_3))).
 \]
Setting $t_1=0.5$ and $t_3=1$, it follows that $\P(\Ec_1^c)\leq n\ex^{-0.09m}$ and $\P(\Ec_5^c)\leq n\ex^{-0.84m}$. From Lemma \ref{lem:singvalues},  $\P(\Ec_4^c)\leq 2\exp(-\frac{m}{2})$. Setting $t_3=1$, Lemma \ref {lem:lowedbound:AD} implies that 
Finally, from Lemma  \ref{lemma:concent-z} combined with the union bound (since $|\Cc|\leq 2^{nr}$) implies that 
\[
\P((\Ec_2\cap\Ec_3)^c)\leq  2^{nr+1}\exp\Big(-{ct_2\min(1,{t_2\over 2m  }) \lambda^2_{\min}(AA^T) x^4_{\min} \over 2 x_{\max}^4\lambda^2_{\max}(A^TA)}\Big).
\]
Let
\[
c_3 \triangleq {c x^4_{\min} (1-2\sqrt{m/n})^4 \over 2 x_{\max}^4(1+2\sqrt{m/n})^4}.
\]
Then, for $t_2<2m$,
\begin{align}
\P((\Ec_2\cap\Ec_3)^c\cup\Ec_4)&\leq  2^{nr+1}\exp(-{c_3t_2^2/2m})\nonumber\\
&=\exp((nr+1) \log 2-{c_3t_2^2/2m}).
\end{align}
Let $t_2=\sqrt{(2\log 2)mnr (1+\epsilon)/c_3}$. Then,
\begin{align}
\P((\Ec_2\cap\Ec_3)^c\cup\Ec_4)&\leq \exp((nr+1) \log 2- (\log 2)nr (1+\epsilon))\nonumber\\
&= 2^{-nr \epsilon+1}.
\end{align}
Let $c_4=c_3/(8\log 2)$.  Since by assumption ${1\over n}\|\xv_o-{\xvt_o}\|^2\leq \delta$, we have ${1\over n}\|\xv_o^2-\xvt_o^2\|^2\leq 4x_{\max}^2\delta$. Also, $\|\xvh_o^2-\xv_o^2\|^2\geq 4x_{\min}^2\|\xvh_o-\xv_o\|^2$. Therefore,  from \eqref{eq:step3-main} it follows  that 
\begin{align}
 {c_1x^2_{\min} \over (1+\lambda_{m})^2} {1\over n} \| \xv_o- \xvh_o\|_2^2 \leq   \sqrt{(1+\epsilon) nr  \over c_4 m} +4 c_2 x^2_{\max}\delta.
\end{align}

\section{Conclusions}\label{sec:conc}

In this paper, we have studied compressed sensing recovery of structured signals  in the presence of speckle noise. In a compressed sensing system where every input pixel is distorted independently by a multiplicative Gaussian noise, we have derived a ML-based recovery method. We have used lossy  compression codes to model the structures of sources. We have shown that given sufficient number of measurements the ML-based recovery method is able to recover a signal from its under-sampled measurements, even in the presence of speckle  noise. To the best of our knowledge, this is the first theoretical result on estimation in the presence of speckle noise. The  ML-based optimization is computationally intractable and cannot be implemented. We have proposed employing projected gradient descent to approximate its solution. Our simulations results show the effectiveness of the proposed method.

\appendix

First note that $W= \sum_{i=1}^n w_i^2$ has a $\chi^2$ distribution $n$ degrees of freedom. Hence, its distribution is given by
\[
f_W(w) = \frac{w^{\frac{n}{2}-1} \exp(-\frac{n}{2}) }{2^{\frac{n}{2}} \Gamma (\frac{n}{2})}.  
\]
Hence, $\E (W^{\frac{1}{2}}) = \frac{\Gamma (\frac{n+1}{2}) \sqrt{2}}{\Gamma (\frac{n}{2})}$. Using the Stirling's formula for the $\Gamma$ function, we have
\begin{eqnarray*}
\log(\Gamma (\frac{n+1}{2}))- \log(\Gamma (\frac{n}{2})) &=& \frac{n+1}{2} \log (\frac{n+1}{2}) - \frac{n+1}{2} + \frac{1}{2} \log(\frac{4 \pi }{n+1})- \frac{n}{2} \log (\frac{n}{2}) + \frac{n}{2} - \frac{1}{2} \log \frac{4 \pi}{n} + O(\frac{1}{n^2}) \nonumber \\
&=& \frac{1}{2}\log (\frac{n+1}{2}) + \frac{n}{2}\log (\frac{n+1}{n})- \frac{1}{2} -\frac{1}{2} \log(\frac{n+1}{n}) + O(\frac{1}{n^2}) \nonumber \\
&=& \frac{1}{2} \log \frac{n}{2} + \frac{1}{2n} + \frac{n}{2} (\frac{1}{n} - \frac{1}{2 n^2}) - \frac{1}{2} - \frac{1}{2n} +O(\frac{1}{n^2}).  \nonumber \\
&=&   \frac{1}{2} \log \frac{n}{2} -\frac{1}{4n} + O(\frac{1}{n^2}). 
\end{eqnarray*}
Hence,
\begin{eqnarray*}
\E (W^{\frac{1}{2}}) = \sqrt{n} (1- \frac{1}{4n}) + o(\frac{1}{\sqrt{n}}). 
\end{eqnarray*}
Therefore,
\begin{eqnarray*}
\E \|\xvh_{\rm ML}- \xv\|_2^2= a^2 n \E ( ({1\over n}\sum_{i=1}^nw_i^2)^{1\over 2})-1)^2 = a^2 n (2- 2 \E  ({1\over n}\sum_{i=1}^nw_i^2)^{1\over 2}) = \frac{a^2}{2}. 
\end{eqnarray*}

\bibliographystyle{unsrt}
\bibliography{myrefs,references}

\begin{thebibliography}{10}

\bibitem{moreira2013tutorial}
A.~{Moreira}, P.~{Prats-Iraola}, M.~{Younis}, G.~{Krieger}, I.~{Hajnsek}, and
  K.~P. {Papathanassiou}.
\newblock A tutorial on synthetic aperture radar.
\newblock {\em IEEE Geo. and Rem. Sen. Mag.}, 1(1):6--43, 2013.

\bibitem{huang1991optical}
D.~Huang, E.~A. Swanson, C.~P. Lin, J.~S. Schuman, W.~G. Stinson, W.~Chang,
  M.~R. Hee, T.~Flotte, K.~Gregory, C.~A. Puliafito, et~al.
\newblock Optical coherence tomography.
\newblock {\em Science}, 254(5035):1178--1181, 1991.

\bibitem{argenti2013tutorial}
F.~Argenti, A.~Lapini, T.~Bianchi, and L.~Alparone.
\newblock A tutorial on speckle reduction in synthetic aperture radar images.
\newblock {\em IEEE Geo. and Rem. Sen. Mag.}, 1(3):6--35, Sep. 2013.

\bibitem{touzi2002review}
R.~{Touzi}.
\newblock A review of speckle filtering in the context of estimation theory.
\newblock {\em IEEE Trans. on Geo. and Rem. Sen.}, 40(11):2392--2404, 2002.

\bibitem{ozcan2016despeckling}
C.~{Ozcan}, B.~{Sen}, and F.~{Nar}.
\newblock Sparsity-driven despeckling for {SAR} images.
\newblock {\em IEEE Geo. and Rem. Sen. Letters}, 13(1):115--119, 2016.

\bibitem{deledalle2014patch}
C.~{Deledalle}, L.~{Denis}, G.~{Poggi}, F.~{Tupin}, and L.~{Verdoliva}.
\newblock Exploiting patch similarity for {SAR}s image processing: The nonlocal
  paradigm.
\newblock {\em IEEE Sig. Proc. Mag.}, 31(4):69--78, 2014.

\bibitem{dimartino2016nonlocal}
G.~{Di Martino}, A.~{Di Simone}, A.~{Iodice}, and D.~{Riccio}.
\newblock Scattering-based nonlocal means {SAR} despeckling.
\newblock {\em IEEE Trans. on Geo. and Rem. Sen.}, 54(6):3574--3588, 2016.

\bibitem{chierchial2017convolutional}
G.~{Chierchia}, D.~{Cozzolino}, G.~{Poggi}, and L.~{Verdoliva}.
\newblock {SAR} image despeckling through convolutional neural networks.
\newblock In {\em 2017 IEEE Int. Geo. and Rem. Sen. Symp. (IGARSS)}, pages
  5438--5441, 2017.

\bibitem{wang2017convolutional}
P.~{Wang}, H.~{Zhang}, and V.~M. {Patel}.
\newblock {SAR} image despeckling using a convolutional neural network.
\newblock {\em IEEE Sig. Proc. Letters}, 24(12):1763--1767, 2017.

\bibitem{yoon2008cs}
Yeo-Sun Yoon and Moeness~G. Amin.
\newblock {Compressed sensing technique for high-resolution radar imaging}.
\newblock In Ivan Kadar, editor, {\em Signal Processing, Sensor Fusion, and
  Target Recognition XVII}, volume 6968, pages 506 -- 515. International
  Society for Optics and Photonics, SPIE, 2008.

\bibitem{patel2009cs}
V.~M. Patel, G.~R. Easley, D.~M. Healy, and R.~Chellappa.
\newblock Compressed sensing for synthetic aperture radar imaging.
\newblock In {\em IEEE Int. Conf. on Image Proc. (ICIP)}, pages 2141--2144,
  2009.

\bibitem{onhnon2010sparsity}
N.~O. Onhon and M.~Çetin.
\newblock Joint sparsity-driven inversion and model error correction for radar
  imaging.
\newblock In {\em 2010 IEEE Int. Conf. on Acou. Speech and Sig. Pro.}, pages
  1206--1209, 2010.

\bibitem{demirci2013cs}
S.~Demirci and C.~Ozdemir.
\newblock Compressed sensing-based imaging of millimeter-wave {ISAR} data.
\newblock {\em Mic. and Opt. Tech. Letters}, 55(12):2967--2972, 2013.

\bibitem{cheng2018pareto}
P.~Cheng and J.~Zhao.
\newblock Generalised {Pareto} distribution-based {Bayesian} compressed sensing
  inverse synthetic aperture radar imaging.
\newblock {\em IET Radar, Son. \& Nav.}, 12(5):549--556, 2018.

\bibitem{bi2017multifrequency}
D.~{Bi}, Y.~{Xie}, L.~{Ma}, X.~{Li}, X.~{Yang}, and Y.~R. {Zheng}.
\newblock Multifrequency compressed sensing for 2-d near-field synthetic
  aperture radar image reconstruction.
\newblock {\em IEEE Trans. on Inst. and Mea.}, 66(4):777--791, 2017.

\bibitem{cetin2014sparsity}
M.~Cetin, I.~Stojanović, N.~O. Onhon, K.~Varshney, S.~Samadi, W.~C. Karl, and
  A.~S. Willsky.
\newblock Sparsity-driven synthetic aperture radar imaging: Reconstruction,
  autofocusing, moving targets, and compressed sensing.
\newblock {\em IEEE Sig. Proc. Mag.}, 31(4):27--40, 2014.

\bibitem{jalali2016compression}
S.~Jalali and A.~Maleki.
\newblock From compression to compressed sensing.
\newblock {\em Appl. Comp. Harmonic Anal. (ACHA)}, 40(2):352--385, 2016.

\bibitem{beygi2019efficient}
S.~Beygi, S.~Jalali, A.~Maleki, and U.~Mitra.
\newblock An efficient algorithm for compression-based compressed sensing.
\newblock {\em Information and Inference: A Journal of the IMA}, 8(2):343--375,
  2019.

\bibitem{bakhshizadeh2020using}
M.~Bakhshizadeh, A.~Maleki, and S.~Jalali.
\newblock Using black-box compression algorithms for phase retrieval.
\newblock {\em IEEE Trans. Inform. Theory}, 66(12):7978--8001, 2020.

\bibitem{RezagahJ:17}
F.~E. Rezagah, S.~Jalali, E.~Erkip, and H.~V. Poor.
\newblock Compression-based compressed sensing.
\newblock {\em IEEE Trans. Inform. Theory}, 63(10):6735--6752, Oct. 2017.

\bibitem{bickel2009simultaneous}
P.~J. Bickel, Y.~Ritov, and A.~B.s Tsybakov.
\newblock Simultaneous analysis of {Lasso} and {Dantzig} selector.
\newblock {\em The Annals of statistics}, 37(4):1705--1732, 2009.

\bibitem{combettes2005signal}
Patrick~L Combettes and Val{\'e}rie~R Wajs.
\newblock Signal recovery by proximal forward-backward splitting.
\newblock {\em Multiscale Modeling $\&$ Simulation}, 4(4):1168--1200, 2005.

\bibitem{blumensath2009iterative}
T.~Blumensath and M.~E. Davies.
\newblock Iterative hard thresholding for compressed sensing.
\newblock {\em Appl. Comp. Harmonic Anal. (ACHA)}, 27(3):265--274, 2009.

\bibitem{tibshirani2005sparsity}
R.~Tibshirani, M.~Saunders, S.~Rosset, J.~Zhu, and K.~Knight.
\newblock Sparsity and smoothness via the fused {Lasso}.
\newblock {\em J. of the Royal Stat. Soc.: Series B (Stat. Meth.)},
  67(1):91--108, 2005.

\bibitem{Viterbi:67}
A.~Viterbi.
\newblock Error bounds for convolutional codes and an asymptotically optimum
  decoding algorithm.
\newblock {\em IEEE Trans. Inform. Theory}, 13(2):260 -- 269, apr 1967.

\bibitem{optuna}
T.~Akiba, S.~Sano, T.~Yanase, T.~Ohta, and M.~Koyama.
\newblock Optuna: A next-generation hyperparameter optimization framework.
\newblock In {\em Proc. of the 25rd {ACM} {SIGKDD} Int. Conf. on Know. Dis. and
  Data Min.}, 2019.

\bibitem{RudelsonVershinin2010}
M.~Rudelson and R.~Vershynin.
\newblock Non-asymptotic theory of random matrices: extreme singular values.
\newblock In {\em Proc. of the Int. Cong. of Math. 2010 (ICM 2010)}, pages
  1576--1602. World Scientific, 2010.

\bibitem{JalaliM:14-MEP-IT}
S.~Jalali, A.~Maleki, and R.~G. Baraniuk.
\newblock Minimum complexity pursuit for universal compressed sensing.
\newblock {\em IEEE Trans. Inform. Theory}, 60(4):2253--2268, Apr. 2014.

\bibitem{laub2005matrix}
A.~J. Laub.
\newblock {\em Matrix analysis for scientists and engineers}, volume~91.
\newblock Siam, 2005.

\end{thebibliography}

\end{document}